\documentclass[12pt]{article}
 
\usepackage{microtype}

\usepackage[margin=1in]{geometry}
\usepackage{amsmath}
\usepackage{amsfonts}
\usepackage{tikz}
\usepackage[utf8]{inputenc}
\usepackage[english]{babel}
\usepackage{amssymb}
\usepackage{url,hyperref}
\usepackage{wrapfig}
\usepackage{graphicx}
\usepackage{caption}
\usepackage{lineno}
\usepackage{subcaption}
\usepackage{amsthm}

\usepackage{float}

\usepackage{algpseudocode}
\usepackage[linesnumbered,ruled,vlined]{algorithm2e}
\usepackage[]{algorithm2e}

\newcommand\ForAuthors[1]
 {\par\smallskip                     
  \begin{center}
   \fbox
   {\parbox{0.9\linewidth}
    {\raggedright\sc--- #1}
   }
  \end{center}
  \par\smallskip                     %
 }

\newcommand\dPk{d_{P,k}}
\newcommand\dPkk{d_{P,k}}
\newcommand\dPkkk{d_{P,2k}}

\newcommand\dPkkp{d_{P,k'}}
\newcommand\dPkpkp{d_{P',k'}}

\newcommand\R{\mathbb{R}}

\newcommand\X{\mathbb{X}}

\newcommand\dX[2]{d_\X(#1,#2)}

\newcommand\Dgm[1]{\mathrm{Dgm}(#1)}
\newcommand\wfs{\mathrm{wfs}}

\newcommand\ensmp[1]{$#1$-noisy sample}
\newcommand\uensmp[2]{uniform $(#1,#2)$-noisy sample}
\newcommand\luensmp[2]{uniform $(#1,#2)$-adaptive noisy sample}
\newcommand\aensmp[1]{$#1$-adaptive noisy sample}
\newcommand\wuensmp[2]{weak uniform $(#1,#2)$-noisy sample} 

\newtheorem{theorem}{Theorem}[section]
\newtheorem{definition}[theorem]{Definition}
\newtheorem{lemma}[theorem]{Lemma}
\newtheorem{claim}[theorem]{Claim}
\newtheorem{corollary}[theorem]{Corollary}
\newtheorem{proposition}[theorem]{Proposition}

\newcommand{\outputP}		{Q}

\newcommand{\intermedP}      	{Q}
\definecolor{darkred}{rgb}{1, 0.1, 0.3}
\definecolor{darkblue}{rgb}{0.1, 0.1, 1}

\newcommand{\denselist}{\itemsep 0pt\parsep=1pt\partopsep 0pt}
\newcommand{\adH}			{\delta_H^f}

\newcommand{\myparagraph}[1]	{{\vspace*{0.1in}\noindent{\bf #1~}}}

\newcounter{magicrownumbers}
\newcommand\rownumber{\stepcounter{magicrownumbers}\arabic{magicrownumbers}}

\title{Declutter and Resample: Towards parameter free denoising}
\author{Micka\"{e}l Buchet\thanks{Advanced Institute for Materials Research, Tohoku University. \texttt{mickael.buchet@m4x.org}}, Tamal K. Dey\thanks{Department of Computer Science and Engineering, The Ohio State University. \texttt{tamaldey, yusu@cse.ohio-state.edu, wang.6195@osu.edu}}, Jiayuan Wang\footnotemark[2], Yusu Wang\footnotemark[2]}

\date{}

\begin{document}

\maketitle

\begin{abstract}
In many data analysis applications the following scenario is commonplace:
we are given a point set that is supposed to sample a hidden ground truth
$K$ in a metric space,
but it got corrupted with noise so that some of the 
data points lie far away from $K$ creating 
outliers also termed as {\em ambient noise}. One of the main
goals of denoising algorithms is to eliminate such noise 
so that the curated data lie within a bounded Hausdorff distance
of $K$. Popular denoising approaches such as deconvolution and thresholding 
often require the user to set several parameters and/or to choose an appropriate
noise model while guaranteeing only asymptotic convergence. 
Our goal is to lighten this burden as much as possible while
ensuring theoretical guarantees in all cases.

Specifically, first, we propose a simple denoising algorithm that requires only
a single parameter but provides a theoretical guarantee on the quality of the output on general input points. 
We argue that this single parameter cannot be avoided. 
We next present 
a simple algorithm that avoids even this parameter by paying
for it with a slight strengthening of the sampling condition on the input points which is not unrealistic.
We also provide some preliminary empirical evidence that our algorithms
are effective in practice. 
\end{abstract}

\section{Introduction}
\label{sec:intro}

Real life data are almost always corrupted by noise.
Of course, when we talk about noise, we
implicitly assume that the data sample
a hidden space called the {\em ground truth} with respect
to which we measure the extent and type of noise.
Some data can lie far away from the ground truth
leading to ambient noise. Clearly, the data density needs to be higher 
near the ground truth if signal has to prevail over noise. 
Therefore, a worthwhile goal of a denoising algorithm is to curate the data, 
eliminating the ambient noise while retaining most of the subset 
that lies within a bounded distance from the ground truth.
 
In this paper we are interested in removing ``outlier''-type of noise from input data. 
Numerous algorithms have been developed for this problem in many different application fields; see e.g \cite{hodge2004survey,zhang2013advancements}. 
There are two popular families of denoising/outlier detection approaches: 
Deconvolution and Thresholding.
Deconvolution methods rely on the knowledge of a generative noise model for the data. 
For example, the algorithm may assume that the input data has 
been sampled according to a probability measure obtained by convolving 
a distribution such as a Gaussian~\cite{dpinsM} with a measure 
whose support is the ground truth. 
Alternatively, it may assume that the data is generated according 
to a probability measure with a small Wasserstein distance to a 
measure supported by the ground truth~\cite{dwmgiCCDM}.
The denoising algorithm attempts to cancel the noise by 
deconvolving the data with the assumed model.  


A deconvolution algorithm requires the knowledge of the generative model 
and at least a bound on the parameter(s) involved, such as the standard deviation of 
the Gaussian convolution or the Wasserstein distance.
Therefore, it requires at least one parameter as well as the knowledge of the noise type.
The results obtained in this setting are often asymptotic, that is,
theoretical guarantees hold in the limit when the number of
points tends to infinity. 

The method of thresholding relies on a density estimation 
procedure~\cite{desdaS} by which it estimates
the density of the input locally. 
The data is cleaned, either by
removing points where density is lower than a threshold~\cite{dstD},
or moving them from such areas toward higher densities using
gradient-like methods such as mean-shift~\cite{msCM,Ozertem11}.
It has been recently used for uncovering geometric information such as 
one dimensional features~\cite{pdgfGPVW}.
In~\cite{MEscalar}, the \emph{distance to a measure}~\cite{gipmCCM} that can also be seen as a density estimator~\cite{wknndegiBCCDR} 
has been exploited for thresholding.
Other than selecting a threshold, 
these methods require the choice of a density estimator.
This estimation requires at least one additional parameter, either to define a kernel, or a mass to define the distance to a measure.
In the case of a gradient based movement of the points, the nature of the 
movement also has to be defined by fixing the length of a step 
and by determining the terminating condition of the algorithm.

\myparagraph{New work.} 
In above classical methods, the user is burdened with making
several choices such as fixing an appropriate noise model, selecting
a threshold and/or other parameters.
Our main goal is to lighten this burden as much as possible.
First, we show that denoising with a single
parameter is possible and this parameter is in some sense unavoidable
unless a stronger sampling condition on the input points is assumed. 
This leads to our main 
algorithm that is completely free of any parameter
when the input satisfies a stronger sampling condition which 
is not unrealistic. 

Our first algorithm \emph{Declutter algorithm} uses a single parameter 
(presented in Section \ref{sec:decluttering}) 
and assumes a very general sampling condition which is
not stricter than those for the classical noise models mentioned
previously because it holds with high probability for
those models as well. 
Additionally, our sampling condition also allows ambient noise and 
locally adaptive samplings.
Interestingly, we note that our Declutter algorithm is in fact a variant of the approach proposed in \cite{chan2006spanners} to construct the so-called {\it $\varepsilon$-density net}. Indeed, as we point out in Appendix D 
, the procedure of \cite{chan2006spanners} can also be directly used for denoising purpose and one can obtain an analog of Theorems \ref{th:sparsifyingGuarantees} and \ref{thm:adaptiveHaus} in this paper for the resulting density net. 

\begin{figure}[!ht]
\centering
\includegraphics[height=.2\textwidth]{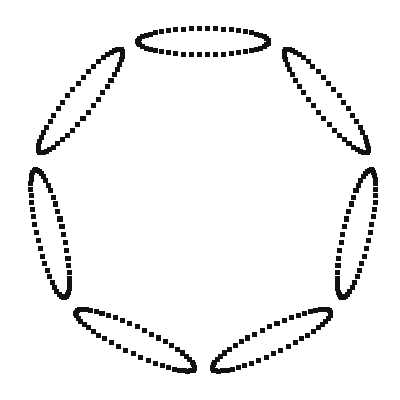}
\includegraphics[height=.2\textwidth]{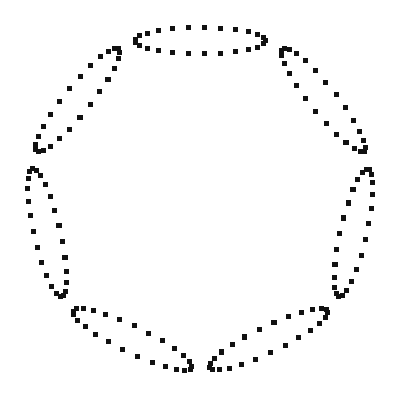}
\includegraphics[height=.2\textwidth]{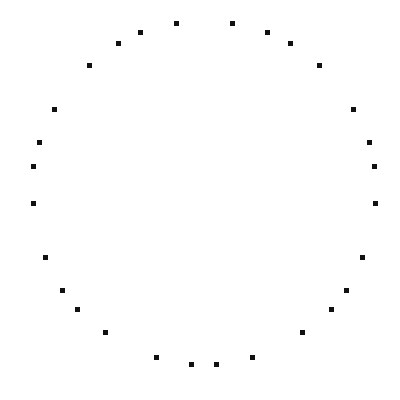}
\includegraphics[height=.2\textwidth]{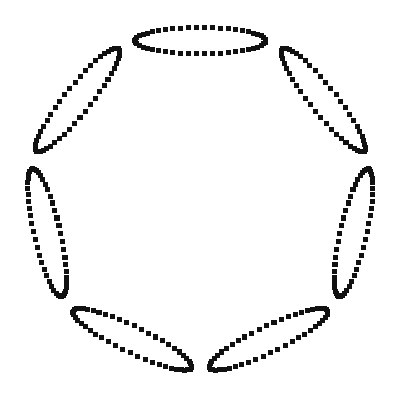}
\vspace*{-0.15in}\caption{{
{\small From left to right: the input sample,the output of Algorithm {\sf Declutter} when $k=2$,  the output of Algorithm {\sf Declutter} when $k=10$,the output of Algorithm {\sf ParfreeDeclutter}. }}\label{fig:scalenecessary}}
\end{figure}

Use of a parameter in the denoising process is unavoidable in some sense, unless there are other assumptions about the hidden space. 
This is illustrated by the example
in Figure~\ref{fig:scalenecessary}.
Does the sample here represent a set of small loops or one big circle?
The answer depends on the scale at which we examine the data.
The choice of a parameter may represent this choice of the scale \cite{boissonnat2009manifold, DGGZ03}.
To remove this parameter, one needs other conditions for either the hidden space itself or for the sample, say by assuming that the data has some uniformity. 
Aiming to keep the sampling restrictions as minimal  
as possible, we show that it is sufficient 
to assume the homogeneity in data \emph{only on or close to} the ground truth
for our second algorithm which requires no input parameter.

Specifically, the parameter-free algorithm presented in Section \ref{sec:parafree} 
relies on an iteration that intertwines our decluttering algorithm 
with a novel resampling procedure.
Assuming that the sample is sufficiently dense and somewhat uniform near the ground truth
at scales beyond a particular scale $s$, our algorithm selects a subset of 
the input point set that is close to the ground truth 
without requiring any input from the user. 
The output maintains the quality at scale $s$ even though the algorithm has no explicit knowledge of this parameter. See Figure \ref{fig:multiscale_square} for an example. 

\begin{figure}[!ht]
\centering
\includegraphics[height=.24\textwidth]{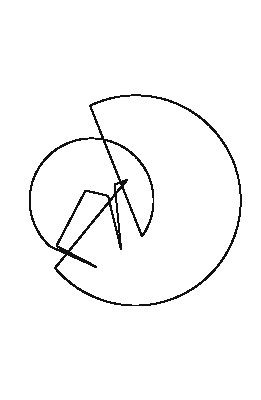}
\includegraphics[height=.24\textwidth]{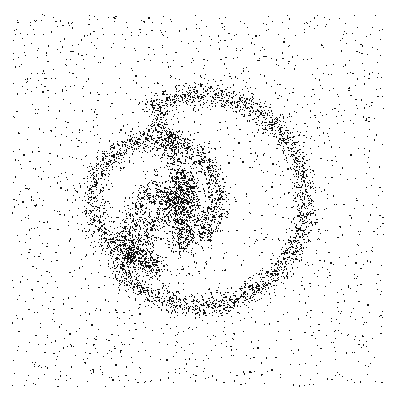}
\includegraphics[height=.24\textwidth]{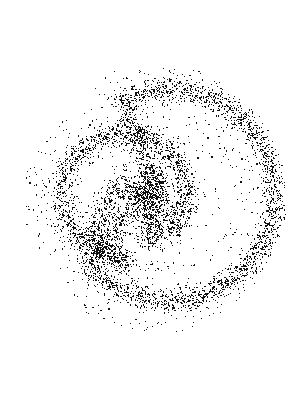}
\includegraphics[height=.24\textwidth]{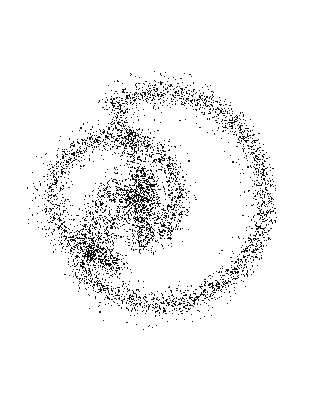}
\includegraphics[height=.24\textwidth]{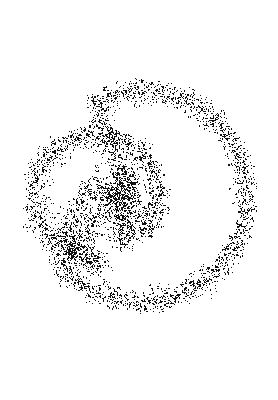}
\vspace*{-0.15in}\caption{{
{\small From left to right: the ground truth, the noisy input samples ($\sim 7000$ points around the ground truth and 2000 ambient noise points), two intermediate steps of our iterative parameter-free denoising algorithm and the final output.}}\label{fig:multiscale_square}}
\end{figure}

All missing details from this extended abstract can be found in the appendix. In addition, in Appendix C, 
we show how the denoised data set can be used for homology inference. In Appendix E, we provide various preliminary experimental results of our denoising algorithms. 

\myparagraph{Remark.} 
Very recently, Jiang and Kpotufe proposed a consistent algorithm for estimating the so-called modal-sets with also only one parameter \cite{jiang2016modal}. The problem setup and goals are very different: In their work, they assume that input points are sampled from a density field that is locally maximal and constant on a compact domain. The goal is to show that as the number of samples $n$ tends to infinity, such domains (referred to as modal-sets in their paper) can be recovered, and the recovered set converges to the true modal-sets under the Hausdorff distance. We also note that our Declutter algorithm allows adaptive sampling as well.  

\section{Preliminaries}\label{sec:preliminaries}
 
We assume that the input is a set of points $P$ sampled around a hidden compact set $K$, the ground truth, in a metric space $(\X, d_\X)$. For simplicity, in what follows the reader can assume $\X = \mathbb{R}^d$ with $P, K \subset \X = \mathbb{R}^d$, and the metric $d_\X$ of $\X$ is simply the Euclidean distance. 
Our goal is to process $P$ into another point set $\outputP$ guaranteed to be Hausdorff close to $K$ and hence to be a better sample of the hidden space $K$ for further applications.
By Hausdorff close, we mean that the (standard) \emph{Hausdorff distance} $\delta_H (\outputP, K)$ between $\outputP$ and $K$, defined as the infimum of $\delta$ such that $\forall p\in \outputP, d_\X (p, K) \le \delta$ and $\forall x \in K, d_\X(x, Q) \le \delta$, is bounded.
Note that ambient noise/outliers can incur a very large Hausdorff distance. 

The quality of the output point set $\outputP$ obviously
depends on the ``quality'' of input points $P$, which we 
formalize via the language of \emph{sampling conditions}.
We wish to produce good quality output for inputs satisfying 
much weaker sampling conditions than a bounded Hausdorff distance. 
Our sampling condition is based on the sampling condition introduced and studied in~\cite{MEthesis,MEscalar}; see Chapter 6 of \cite{MEthesis} for discussions on the relation of their sampling condition with some of the common noise models such as Gaussian. 
Below, we first introduce a basic sampling condition deduced from the
one in \cite{MEthesis,MEscalar}, and then introduce its extensions incorporating adaptivity and uniformity.  

\myparagraph{Basic sampling condition.}
Our sampling condition is built upon the concept of $k$-distance, 
which is a specific instance of a broader concept
called \emph{distance to a measure} introduced in~\cite{gipmCCM}. 
The $k$-distance $\dPkk(x)$ is simply the root mean of square distance from $x$ to 
its $k$-nearest neighbors in $P$. The averaging makes it robust to outliers. 
One can view $\dPkk(x)$ as capturing the inverse of the 
density of points in $P$ around $x$ \cite{wknndegiBCCDR}.  
As we show in Appendix D
, this specific form of $k$-distance is not essential -- Indeed, several of its variants can replace its role in the definition of sampling conditions below, and our Declutter algorithm will achieve similar denoising guarantees. 

\begin{definition}[\cite{gipmCCM}]\label{def:kdist}
Given a point $x \in \X$, let $p_i(x) \in P$ denote the $i$-th nearest neighbor of $x$ in $P$. The \emph{$k$-distance} to a point set $P\subseteq \X$ is 
$\dPkk(x)=\sqrt{\frac{1}{k}\sum_{i=1}^k\dX{x}{p_i(x)}^2}$.
\end{definition}
\begin{claim}[\cite{gipmCCM}]\label{claim:dPkkLip}
$\dPkk(\cdot)$ is 1-Lipschitz, 
i.e. $|\dPkk(x)-\dPkk(y)|\leq \dX{x}{y}$ for $\forall (x,y)\in \X\times\X$.
\end{claim}
All our sampling conditions are dependent on 
the choice of $k$ in the $k$-distance, which we reflect
by writing $\epsilon_k$ instead of $\epsilon$ in the sampling conditions below. 
The following definition is related to the sampling condition proposed in \cite{MEscalar}.

\begin{definition}\label{def:sampling}
Given a compact set $K\subseteq\X$ 
and a parameter $k$, a point set $P$ is an \emph{\ensmp{\epsilon_k}} of $K$ if
\begin{enumerate}\denselist
\item $\forall x\in K,\ \dPkk(x)\leq\epsilon_k$
\item $\forall x\in\X,\ \dX{x}{K}\leq\dPkk(x)+\epsilon_k$
\end{enumerate}
\end{definition}
 
Condition $1$ in Definition \ref{def:sampling} means that the density of $P$ on the compact set $K$ is bounded
from below, that is, $K$ is well-sampled by $P$. Note, we only require $P$ to be a dense enough sample of $K$ -- there is no uniformity requirement in the sampling here. 

Condition $2$ implies that a point with low $k$-distance, 
i.e. lying in high density region, has to be close to $K$. 
Intuitively, $P$ can contain outliers which can form small clusters 
but their density can not be significant compared to the density of 
points near the compact set $K$. 

Note that the choice of $\epsilon_k$ always exists for a bounded point set $P$, no matter what value of $k$ is -- For example, one can set $\epsilon_k$ to be the diameter of point set $P$. However, the smallest possible choice of $\epsilon_k$ to make $P$ an $\epsilon_k$-noisy sample of $K$ depends on the value of $k$. We thus use $\epsilon_k$ in the sampling condition to reflect this dependency. 

%

In Section~\ref{sec:parafree}, we develop  
a parameter-free denoising algorithm. 
As Figure~\ref{fig:scalenecessary}
illustrates, it is necessary to have a 
mechanism to remove potential ambiguity about the ground truth. 
We do so by using a stronger sampling condition to enforce some degree of uniformity: 

\begin{definition}\label{def:unisample}
Given a compact set $K\subseteq \X$ and a parameter $k$, a point set $P$ is a \uensmp{\epsilon_k}{c} of $K$ if $P$ is an \ensmp{\epsilon_k} of $K$ (i.e, conditions of Def. \ref{def:sampling} hold) and
\begin{enumerate}
\setcounter{enumi}{2}
\item $\forall p\in P,\ \dPkk(p)\geq\frac{\epsilon_k}{c}.$
\end{enumerate}
\end{definition}


It is important to note that the lower bound in Condition $3$ enforces that the sampling needs to be homogeneous -- i.e, $\dPkk(x)$ is bounded both from above and from below by some constant factor of $\epsilon_k$ -- \emph{only for points on and around} the ground truth $K$. This is because condition $1$ in Def. \ref{def:sampling} is only for points from $K$, and condition 1 together with the $1$-Lipschitz property
of $\dPk$ (Claim \ref{claim:dPkkLip}) leads to an upper bound of $O(\epsilon_k)$ for $\dPkk(y)$ only
for points $y$ within $O(\epsilon_k)$ distance to $K$. 
There is no such upper bound on $\dPk$ for noisy points far away from $K$ and \emph{thus no 
homogeneity/uniformity requirements for them}. 

\myparagraph{Adaptive sampling conditions.}
The sampling conditions given above are global, meaning that
they do not respect the ``features" of the ground truth.
We now introduce an adaptive version of the sampling conditions
with respect to a 
feature size function. 

\begin{definition}\label{def:lfs}
Given a compact set $K\subseteq\X$, a feature size function $f:K\rightarrow \mathbb{R}^+\cup\{0\}$ is a $1$-Lipschitz non-negative real function on $K$.
\end{definition}

\noindent Several feature sizes exist in the
literature of manifold reconstruction and topology
inference, including the \emph{local feature size}~\cite{AB99}, 
\emph{local weak feature size}, \emph{$\mu$-local weak feature size}~\cite{stcsesCCL} or \emph{lean set feature size}~\cite{sdlutimDDW}.
All of these functions describe how complicated a compact set 
is locally, and therefore indicate how dense a sample 
should be locally so that information can be inferred faithfully.
Any of these functions can be used as a feature size function 
to define the adaptive sampling below. 
Let $\bar p$ denote any one of the nearest points of $p$ in $K$.
Observe that, in general, a point $p$ can have multiple such nearest
points.

\begin{definition}\label{def:adasample}
Given a compact set $K\subseteq\X$, a feature size function $f$ of $K$, and a parameter $k$, a point set $P$ is a \emph{\luensmp{\epsilon_k}{c} of $K$} if
\begin{enumerate}\denselist
\item$\forall x\in K,\ \dPk(x)\leq \epsilon_k f(x)$.
\item$\forall y\in\X,\ \dX{y}{K}\leq\dPk(y)+\epsilon_k f(\bar y)$.
\item $\forall p\in P,\ \dPk(p)\geq \frac{\epsilon_k}{c} f(\bar p)$.
\end{enumerate}
We say that $P$ is \emph{an \aensmp{\epsilon_k} of $K$} if only conditions 1 and 2 above hold. 
\end{definition}
We require that the feature size is {\em positive everywhere} as otherwise, the sampling condition may require infinite samples in some cases. 
We also note that the requirement of the feature size function being 1-Lipschitz is only needed to provide the theoretical guarantee for our second parameter-free algorithm. 


\section{Decluttering}
\label{sec:decluttering}

We now present a simple yet effective denoising algorithm which takes as input a set of points $P$ and a parameter $k$, and outputs a set of points $\outputP \subseteq P$ with the following guarantees: If $P$ is an \ensmp{\epsilon_k} of a 
hidden compact set $K\subseteq \X$, then the output $\outputP$ lies close to $K$ in the \emph{Hausdorff distance} (i.e, within a small tubular neighborhood of $K$ and outliers are all eliminated). 
The theoretical guarantee holds for both the non-adaptive and the adaptive cases, as stated in Theorems \ref{th:sparsifyingGuarantees} and \ref{thm:adaptiveHaus}. 
  
\RestyleAlgo{boxruled}
\LinesNumbered
\begin{algorithm}[hptb]
\caption{{\sf Declutter}($P$,$k$) \label{alg:sparsification}}
\DontPrintSemicolon
\KwData{Point set $P$, parameter $k$}
\KwResult{Denoised point set $\outputP$}
\Begin{
sort P such that $\dPkk(p_1)\leq\dots\leq\dPkk(p_{|P|})$.\;

$Q_0\longleftarrow\emptyset$ \;

\For{$i\longleftarrow 1$ to $|P|$}{
\If{$Q_{i-1}\cap B(p_i,2\dPkk(p_i))=\emptyset$}{$Q_i=Q_{i-1}\cup \{p_i\}$}
{\bf else}~~ $Q_i=Q_{i-1}$
}
$Q\longleftarrow Q_n$\;
}
\end{algorithm}

As the $k$-distance behaves like the inverse of density,
points with a low $k$-distance are expected
to lie close to the ground truth $K$.
A possible approach is to fix a threshold $\alpha$ 
and only keep the points with a $k$-distance less than $\alpha$.
This thresholding 
solution requires an additional parameter $\alpha$. 
Furthermore, very importantly, such a thresholding approach does not easily work for adaptive samples, where the density in an area with large feature size can be lower than the density 
of noise close to an area with small feature size.

Our algorithm {\sf Declutter}($P$,$k$), presented in {\bf Algorithm 1}, works around these problems by considering 
the points in the order of increasing values of their $k$-distances and using a pruning step: 
Given a point $p_i$, if there exists a point $q$ deemed better in its vicinity, i.e., $q$ has smaller $k$-distance and has been previously selected $(q\in Q_{i-1})$, then $p_i$ is not necessary to describe the ground truth 
and we throw it away.
Conversely, if no point close to $p_i$ has already been selected, 
then $p_i$ is meaningful and we keep it.
The notion of ``closeness'' or ``vicinity'' is defined using the $k$-distance, 
so $k$ is the only parameter.
In particular, the ``vicinity" of a point $p_i$ is defined as the metric ball $B(p_i, 2\dPkk(p_i))$; observe that
this radius is different for different points, and the radius of the ball is larger for outliers. 
Intuitively, the radius $2\dPkk(p_i)$ of the ``vicinity'' around $p_i$  
can be viewed as the length we have to go over to reach the hidden domain with certainty. So, bad points have a larger ``vicinity''.
We remark that this process is related to the construction of the ``density net'' introduced in \cite{chan2006spanners}, which we discuss more in Appendix D 
. 

\begin{wrapfigure}{r}{0.27\textwidth}
\vspace*{-0.15in}
\begin{center}
\includegraphics[width=.30\textwidth]{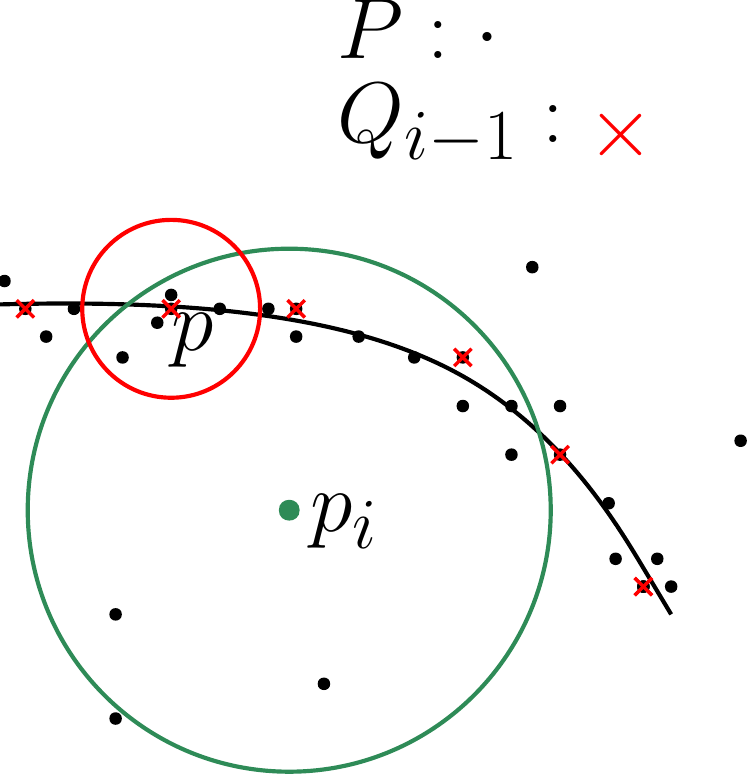}
\end{center}
\vspace*{-0.15in}
\caption{Declutter.
\label{fig:declutter_algorithm}}
\vspace*{-0.15in}
\end{wrapfigure}
See Figure \ref{fig:declutter_algorithm} on the right for an artificial example, where the black points are input points, and red crosses are in the current output $Q_{i-1}$. Now, at the $i$th iteration, suppose we are processing the point $p_i$ (the green point). Since within the vicinity of $p_i$ there is already a good point $p$, we consider $p_i$ to be not useful, and remove it. 
Intuitively, for an outlier $p_i$, it has a large $k$-distance and hence a large vicinity. As we show later, our $\epsilon_k$-noisy sampling condition ensures that this vicinity of $p_i$ reaches the hidden compact set which the input points presumably sample. Since points around the hidden compact set should have higher density, there should be a good point already chosen in $Q_{i-1}$. 
Finally, it is also important to note that, contrary to many common sparsification procedures, our {\sf Declutter} algorithm removes a noisy point because it has a good point within its vicinity, and {\it not because} it is within the vicinity of a good point. For example, in Figure \ref{fig:declutter_algorithm}, the red points such as $p$ have small vicinity, and $p_i$ is not in the vicinity of any of the red point. 

In what follows, we will make this intuition more concrete. 
We first consider the simpler non-adaptive case where $P$ is an \ensmp{\epsilon_k} of $K$. 
We establish that $\outputP$ and the ground truth $K$ are Hausdorff close in the following two lemmas. 
%
The first lemma says that the ground truth $K$ is well-sampled (w.r.t. $\epsilon_k$) by the output $Q$ of {\sf Declutter}. 
\begin{lemma}\label{lem:globalDensity}
Let $Q\subseteq P$ be the output of {\sf Declutter}($P$,$k$) where $P$
is an \ensmp{\epsilon_k} of a compact set $K\subseteq \X$. Then,
for any $x\in K$, there exists $q\in \outputP$ such that $\dX{x}{q}\leq  5\epsilon_k$.
\end{lemma}
\begin{proof}
Let $x\in K$. By Condition 1 of Def. \ref{def:sampling}, we have $\dPkk(x) \le \epsilon_k$. This means that the nearest neighbor $p_i$ of $x$ in $P$ satisfies $\dX{p_i}{x}\leq\dPkk(x)\leq\epsilon_k$. 
If $p_i\in \outputP$, then the claim holds by setting $q = p_i$.
If $p_i \notin \outputP$, there must exist $j < i$ with $p_j \in \outputP_{i-1}$ such that $\dX{p_i}{p_j}\leq 2\dPkk(p_i)$. In other words, $p_i$ was removed by our algorithm because $p_j \in Q_{i-1} \cap B(p_i, 2\dPkk(p_i))$. 
Combining triangle inequality with the 1-Lipschitz property of $\dPkk$ (Claim \ref{claim:dPkkLip}), we then have 
$$\dX{x}{p_j} \le \dX{x}{p_i} + \dX{p_i}{p_j}\leq \dX{x}{p_i}+2\dPkk(p_i)\leq 2\dPkk(x)+3\dX{p_i}{x} \le 5\epsilon_k,$$ 
which proves the claim. 
\end{proof}

The next lemma implies that all outliers are removed by our denoising algorithm. 
\begin{lemma}\label{lemma:QtoK}
Let $Q\subseteq P$ be the output of {\sf Declutter}($P$,$k$) where $P$
is an \ensmp{\epsilon_k} of a compact set $K\subseteq \X$. Then,
for any $q\in \outputP$, there exists $x\in K$ such that $\dX{q}{x}\leq 7\epsilon_k.$ 
\end{lemma}
\begin{proof}
Consider any $p_i\in P$ and let $\bar p_i$ be one of its nearest
points in $K$.
It is sufficient to show that if
$\dX{p_i}{\bar p_i}>7\epsilon_k$, then $p_i\notin \outputP$ . 

Indeed, by Condition 2 of Def. \ref{def:sampling}, $\dPkk(p_i)\geq \dX{p_i}{\bar p_i}-\epsilon_k>6\epsilon_k$.
By Lemma~\ref{lem:globalDensity}, there exists $q\in \outputP$ such that $\dX{\bar p_i}{q}\leq 5\epsilon_k$. Thus,
$$
\dPkk(q)\leq \dPkk(\bar{p}_i)+\dX{\bar p_i}{q} \leq 6\epsilon_k.
$$
Therefore, $\dPkk(p_i)> 6\epsilon_k\geq \dPkk(q)$ implying that $q\in Q_{i-1}$.
Combining triangle inequality and Condition 2 of Def. \ref{def:sampling}, we have 
$$\dX{p_i}{q}\leq \dX{p_i}{\bar p_i}+\dX{\bar p_i}{q}
\leq \dPkk(p_i)+\epsilon_k+5\epsilon_k
< 2\dPkk(p_i). $$
Therefore, $q\in Q_{i-1}\cap B(p_i,2\dPkk(p_i))$, meaning that $p_i\notin \outputP$. 
\end{proof}

\begin{theorem}\label{th:sparsifyingGuarantees}
Given a point set $P$ which is an \ensmp{\epsilon_k} of a compact 
set $K\subseteq\X$, Algorithm 
{\sf Declutter} returns a set $\outputP\subseteq P$ such that
$\delta_H(K,\outputP)\leq 7\epsilon_k.$
\end{theorem}

Interestingly, if the input point set is uniform then the denoised set is 
also uniform, a fact that turns out to be useful for our parameter-free
algorithm later. 
\begin{proposition}
If $P$ is a \uensmp{\epsilon_k}{c} of a compact set $K\subseteq\X$, then the distance between any pair of points of $\outputP$ is at least $2\frac{\epsilon_k}{c}$.
\end{proposition}
\begin{proof}
Let $p$ and $q$ be in $\outputP$ with $p\neq q$ and, 
assume without loss of generality  that $\dPkk(p)\leq\dPkk(q)$.
Then, $p\notin B(q,2\dPkk(q))$ and $\dPkk(q)\geq\frac{\epsilon_k}{c}$.
Therefore, $\dX{p}{q}\geq 2\frac{\epsilon_k}{c}$.
\end{proof}

\paragraph*{Adaptive case.}
Assume the input is an adaptive sample $P\subseteq\X$ with respect to a feature size function $f$.
The denoised point set $\outputP$ may also be adaptive. 
We hence need an adaptive version of the Hausdorff distance denoted $\adH(\outputP, K)$ and defined as the infimum of $\delta$ such that (i) $\forall p \in \outputP$, $d_\X (p, K) \le \delta f(\bar p)$, and (ii) $\forall x \in K, d_\X(x, \outputP) \le \delta f(x)$, where $\bar p$ is a nearest point of $p$ in $K$. 
Similar to the non-adaptive case, we establish that $P$ and output 
$Q$ are Hausdorff close via Lemmas \ref{lem:adaptiveDensity} and \ref{lem:closeness} whose 
 proofs are same as those for Lemmas \ref{lem:globalDensity} and \ref{lemma:QtoK} respectively, but using an adaptive distance w.r.t. the feature size function.  
Note that the algorithm does not need to know what the feature size function $f$ is, hence only one parameter ($k$) remains. 

\begin{lemma}\label{lem:adaptiveDensity}
Let $Q\subseteq P$ be the output of {\sf Declutter}$(P,k)$ where $P$
is an \aensmp{\epsilon_k} of a compact set $K\subseteq \X$. Then,
$\forall x\in K,\exists q\in \outputP,\ \dX{x}{q}\leq  5\epsilon_k f(x) .$
\end{lemma}
%

\begin{lemma}\label{lem:closeness}
Let $Q\subseteq P$ be the output of {\sf Declutter}$(P,k)$ where $P$
is an \aensmp{\epsilon_k} of a compact set $K\subseteq \X$. Then,
for $\forall q\in \outputP,\ \dX{q}{\bar q}\leq 7\epsilon_k f(\bar q). $
\end{lemma}
%
%
%

\begin{theorem}\label{thm:adaptiveHaus}
Given an \aensmp{\epsilon_k} $P$ of a 
compact set $K\subseteq \X$ with
feature size $f$, 
Algorithm {\sf Declutter}
returns a sample $\outputP\subseteq P$ of $K$ where $\adH(\outputP, K) \le 7\epsilon_k$. 
\end{theorem}
%
Again, if the input set is uniform, 
the output remains uniform as stated below. 

\begin{proposition}
Given an input point set $P$ which is an \luensmp{\epsilon_k}{c} of a 
compact set $K\subseteq \X$, the output $Q\subseteq P$ of {\sf Declutter} satisfies
$$\forall (q_i,q_j)\in Q,\ i\neq j\implies \dX{q_i}{q_j}\geq 2\frac{\epsilon_k}{c} f(\bar q_i)$$
\end{proposition}
%

\begin{proof}
Let $q_i$ and $q_j$ be two points of $Q$ with $i<j$.
Then $q_i$ is not in the ball of center $q_j$ and radius $2\dPkk(q_j)$.
Hence $\dX{q_i}{q_j}\geq2\dPkk(q_j)\geq2\frac{\epsilon_k}{c} f(\bar q_j)$.
Since $i < j$, it also follows that $\dX{q_i}{q_j}\geq2\dPkk(q_j)\geq 2\dPkk(q_i)\geq2\frac{\epsilon_k}{c} f(\bar q_i)$. 
\end{proof}

The algorithm {\sf Declutter} removes outliers from the input point set $P$.
As a result, we obtain a point set which lies
close to the ground truth with respect to the Hausdorff distance.
Such point sets can be used for inference about the ground truth
with further processing. 
For example, in topological data analysis, our result can be used 
to perform topology inference from noisy input points 
in the non-adaptive setting; see Appendix C 
  for more details. 


An example of the output of algorithm {\sf Declutter} is given in Figure \ref{fig:decluttering} (a) -- (d). More examples (including for adaptive inputs) can be found in 
Appendix E 
. 

\paragraph*{Extensions.} 
It turns out that there are many choices that can be used for the $k$-distance $\dPkk(x)$ instead of the one introduced in Definition \ref{def:kdist}. 
Indeed, the goal of $k$-distance intuitively is to provide a more robust distance estimate -- Specifically, assume $P$ is a noisy sample of a hidden domain $K \subset \X$. With the presence of noisy points far away from $K$, the distance $d_\X(x, P)$ no longer serves as a good approximation of $d_\X(x, K)$, the distance from $x$ to the hidden domain $K$. 
We thus need a more robust distance estimate. The $k$-distance $\dPkk(x)$ introduced in Definition \ref{def:kdist} is one such choice, and there are many other valid choices. 
As we show in Appendix D
, we only need the choice of $\dPkk(x)$ to be $1$-Lipschitz, and is less sensitive than $d_\X(x, P)$ (that is, $d_\X(x, P) \le \dPkk(x)$). 
We can then define the sampling condition (as in Definitions \ref{def:sampling} and \ref{def:unisample}) using a different choice of $\dPkk(x)$, and Theorems \ref{th:sparsifyingGuarantees} and \ref{thm:adaptiveHaus} still hold. 
For example, we could replace $k$-distance by $\dPkk(x) = \frac{1}{k}\sum_{i=1}^k d(x, p_i(x))$ where $p_i(x)$ is the $i$th nearest neighbor of $x$ in $P$; that is, $\dPkk(x)$ is the average distance to the $k$ nearest neighbors of $x$ in $P$. 
Alternatively, we can replace $k$-distance by $\dPkk(x) = d(x, p_k(x))$, the distance from $x$ to its $k$-th nearest neighbor in $P$ (which was used in \cite{chan2006spanners} to construct the $\varepsilon$-density net). 
{\sf Declutter algorithm} works for all these choices with the same denoising guarantees. 

One can in fact further relax the conditions on $\dPkk(x)$ or even on the input metric space $(\X, d_\X)$ such that the triangle inequality for $d_\X$ only approximately holds. The corresponding guarantees of our Declutter algorithm are provided in Appendix D 
. 

\section{Parameter-free decluttering}
\label{sec:parafree}

The algorithm {\sf Declutter} is not entirely satisfactory.
First, we need to fix the parameter $k$ a priori. Second, while providing a Hausdorff distance guarantee, this procedure also ``sparsifies" input points. Specifically, the empty-ball test also induces some degree of sparsification, as for any point $q$ kept in $\outputP$, the ball $B(q, 2\dPkk(q))$ does not contain any other output points in $\outputP$. 
While this sparsification property is desirable for some applications, it removes too many points in some cases -- See Figure~\ref{fig:decluttering} for an example, where the output density is dominated by $\epsilon_k$ and does not preserve the dense sampling provided by the input around the hidden compact set $K$. In particular, for $k=9$, it does not completely remove ambient noise, while, for $k = 30$, the output is too sparse.

\begin{figure}[!t]
\centering
\begin{tabular}{ccccc}
\includegraphics[height=.23\textwidth]{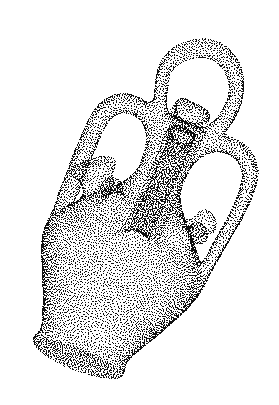} &
\includegraphics[height=.23\textwidth]{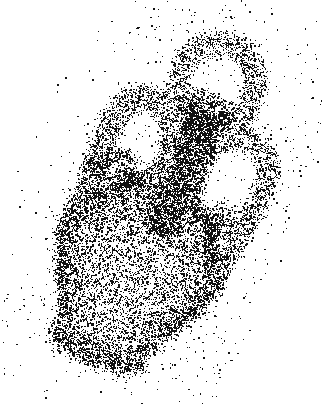} &
\includegraphics[height=.23\textwidth]{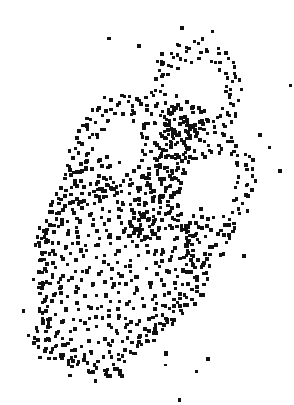} &
\includegraphics[height=.23\textwidth]{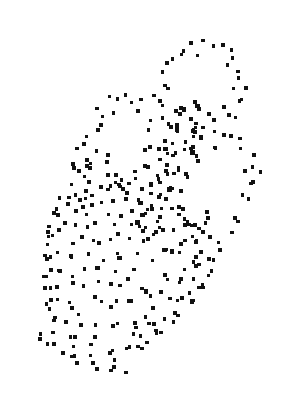} &
\includegraphics[height=.23\textwidth]{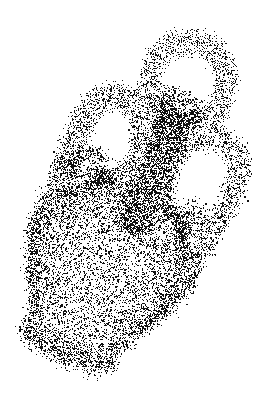}\\
(a) & (b) & (c) & (d) & (e) 
\end{tabular}
\vspace*{-0.2in}\caption{{\small (a) -- (d) show results of the Algorithm {\sf Declutter}: 
(a) the ground truth, (b) the noisy input with 15K points with 1000 ambient noisy points, (c) the output of Algorithm {\sf Declutter} when $k=9$,  (d) the output of Algorithm {\sf Declutter} when $k=30$. In (e), we show the output of Algorithm {\sf ParfreeDeclutter}. As shown in Appendix E
, algorithm {\sf ParfreeDeclutter} can remove ambient noise for much sparser input samples with more noisy points.}
\label{fig:decluttering}}
\end{figure}

\RestyleAlgo{boxruled}
\LinesNumbered
\begin{algorithm}[H]
\label{alg:resampling}
\caption{{\sf ParfreeDeclutter}($P$)}
\DontPrintSemicolon
\KwData{Point set $P$}
\KwResult{Denoised point set $P_0$}
\Begin{
Set $i_* =\lfloor\log_2(|P|)\rfloor$, and $P_{i_*}\longleftarrow P$\;

\For{$i\longleftarrow i_*$ to $1$}{
$\intermedP\longleftarrow$ \sf {Declutter}($P_i$,$2^i$)\;

$P_{i-1}\longleftarrow\cup_{q\in \intermedP}B(q,(10+2\sqrt 2)d_{P_i,2^i}(q))\cap P_i$\;
}}
\end{algorithm}

In this section, we address both of the above concerns by a novel 
iterative re-sampling procedure as described in 
Algorithm {\sf ParfreeDeclutter}($P$).
Roughly speaking, we start with $k = |P|$ and gradually decrease it by halving each time. 
At iteration $i$, let $P_i$ denote the set of points so far kept by the algorithm; $i$ is initialized to be $\lfloor\log_2(|P|)\rfloor$ and 
is gradually decreased. 
We perform the denoising algorithm {\sf Declutter}($P_i, k =2^i$) given in the previous section to first denoise $P_i$  
and obtain a denoised output set $\intermedP$. 
This set can be too sparse. We enrich it by 
re-introducing some points from $P_{i}$, obtaining a 
denser sampling $P_{i-1} \subseteq P_i$ of the ground truth. 
We call this a \emph{re-sampling} process. 
This re-sampling step may bring some outliers back into the current set. 
However, it turns out that a repeated cycle of decluttering and resampling
with decreasing values of $k$ removes these outliers progressively. See Figure \ref{fig:multiscale_square} and also more examples in
Appendix E.
The entire process remains free of any user
supplied parameter. 
In the end, we show that for an input that satisfies a uniform sampling condition, we can obtain an output set which is both dense and Hausdorff close to the hidden compact set, without the need to know the parameters of the input sampling conditions. 

In order to formulate the exact statement of Theorem \ref{thm:resampling}, 
we need to introduce a more relaxed sampling condition. 
We relax the notion of \uensmp{\epsilon_k}{c} by removing condition $2$.
We call it a \emph{\wuensmp{\epsilon_k}{c}}.
Recall that condition $2$ was the one forbidding the noise to be too dense. So essentially, a \wuensmp{\epsilon_k}{c} only concerns points on and around the ground truth, with no conditions on outliers. 

\begin{theorem}\label{thm:resampling}
Given a point set $P$ and $i_0$ such that for all $i>i_0$, $P$ is a \wuensmp{\epsilon_{2^i}}{2} of $K$ and is also a \uensmp{\epsilon_{2^{i_0}}}{2} of $K$, 
Algorithm {\sf ParfreeDeclutter} returns a point set $P_0\subseteq P$ such that $d_H(P_0,K)\leq (87+16\sqrt{2})\epsilon_{2^{i_0}}$.
\end{theorem}

\noindent We elaborate a little on the sampling conditions. 
On one hand, as illustrated by Figure~\ref{fig:scalenecessary}, the uniformity on input points is somewhat necessary in order to obtain a parameter-free algorithm. So requiring  a \uensmp{\epsilon_{2^{i_0}}}{2} of $K$ is reasonable. 
Now it would have been ideal
if the theorem only required that $P$ is a \uensmp{\epsilon_{2^{i_0}}}{2} of $K$ for some $k_0 = 2^{i_0}$. 
However, to make sure that this uniformity is not destroyed during our iterative declutter-resample process before we reach $i = i_0$, we also need to assume that, \emph{around the compact set}, the sampling is uniform for any $k = 2^i$ with $i > i_0$ (i.e, before we reach $i = i_0$). 
The specific statement for this guarantee is given in Lemma \ref{lem:resamplingGuarantees}. 
However, while the uniformity for points \emph{around the compact set} is required for any $i > i_0$, the condition that noisy points 
cannot be arbitrarily dense is \emph{only} required for one parameter, $k = 2^{i_0}$. 

The constant for
the ball radius in the resampling step
is taken as $10+2\sqrt 2$ which we call the resampling
constant $C$. Our
theoretical guarantees hold with this 
resampling constant though a value of $4$ works well in practice. 
The algorithm reduces more noise with decreasing $C$. 
On the flip side, the risk of removing points causing 
loss of true signal also increases with decreasing $C$.  
Section \ref{sec:exp} and Appendix E 
  provide several results for Algorithm {\sf ParfreeDeclutter}. We also point out that while our theoretical guarantee is for non-adaptive case, in practice, the algorithm works well on adaptive sampling as well. 

\paragraph*{Proof for Theorem \ref{thm:resampling}.} 
Aside from the technical Lemma~\ref{lem:iNNupperBound} on the $k$-distance, the proof is divided into three steps.
First, Lemma~\ref{lem:resamplingGuarantees} shows that applying the loop of the algorithm once with parameter $2k$ does not alter the existing sampling conditions for $k'\leq k$.
This implies that the \ensmp{\epsilon_{2^{i_0}}} condition on $P$ will also hold for $P_{i_0}$.
Then Lemma~\ref{lem:resamplingOutliers} guarantees that the step going from $P_{i_0}$ to $P_{i_0-1}$ will remove all outliers.
Combined with Theorem~\ref{th:sparsifyingGuarantees}, which guarantees that $P_{i_0-1}$ samples well $K$, it guarantees that the Hausdorff distance between $P_{i_0-1}$ and $K$ is bounded.
However, we do not know $i_0$ and we have no means 
to stop the algorithm at this point.
Fortunately, we can prove Lemma~\ref{lem:conservationGuaranty} which guarantees that the remaining iterations will not remove too many points and break the theoretical guarantees -- that is, no harm is done in the subsequent iterations even after $i < i_0$. 
Putting all three together leads to our main result Theorem \ref{thm:resampling}.

\begin{lemma}\label{lem:iNNupperBound}
Given a point set $P$, $x\in\X$ and $0\le i \le k$, the distance to the $i$-th nearest neighbor of $x$ in $P$ satisfies, $\dX{x}{p_i}\leq\sqrt{\frac{k}{k-i+1}}\dPkk(x)$.
\end{lemma}
\begin{proof}
The claim is proved by the following derivation.  
$$\frac{k-i+1}{k}\dX{x}{p_i}^2\leq\frac{1}{k}\sum_{j=i}^{k} \dX{x}{p_j}^2\leq \frac{1}{k}\sum_{j=1}^k\dX{x}{p_j}^2=\dPkk(x)^2. $$
\end{proof}

\begin{lemma}\label{lem:resamplingGuarantees}
Let $P$ be a \wuensmp{\epsilon_{2k}}{2} of $K$.
For any $k'\leq k$ such that $P$ is a (weak) \uensmp{\epsilon_{k'}}{c} of $K$ for some $c$, applying one step of the algorithm, with parameter $2k$ 
and resampling constant $C=10+2\sqrt{2}$ 
gives a point set $P'\subseteq P$ which is a (weak) \uensmp{\epsilon_{k'}}{c} of $K$.
\end{lemma}
\begin{proof}
We show that if $P$ is a \uensmp{\epsilon_{k'}}{c}
of $K$, then $P'$ will also be a \uensmp{\epsilon_{k'}}{c}
of $K$. The similar version for weak uniformity follows from the same argument. 

First, it is easy to see that as $P' \subset P$, the second and third sampling conditions of Def. \ref{def:unisample} hold for $P'$ as well. What remains is to show that Condition 1 also holds. 

Take an arbitrary point $x\in K$.
We know that $\dPkkk(x)\leq\epsilon_{2k}$ as $P$ is a \wuensmp{\epsilon_{2k}}{2} of $K$. Hence there exists 
$p\in P$ such that $\dX{p}{x}\leq \dPkkk(x)\leq \epsilon_{2k}$ and $\dPkkk(p)\leq 2\epsilon_{2k}$.
Writing $Q$ the result of the decluttering step, $\exists q\in Q$ such that $\dX{p}{q}\leq 2\dPkkk(p)\leq 4\epsilon_{2k}$.
Moreover, $\dPkkk(q)\geq\frac{\epsilon_{2k}}{2}$ due to the uniformity condition for $P$.

Using Lemma~\ref{lem:iNNupperBound}, for $k' \le k$, the $k'$ nearest neighbors of $x$, which are chosen from $P$, $NN_{k'}(x)$ satisfies:
$$
NN_{k'}(x)\subset B(x,\sqrt{2}\epsilon_{2k})
\subset B(p,(1+\sqrt{2})\epsilon_{2k})
\subset B(q,(5+\sqrt{2})\epsilon_{2k})
\subset B(q,(10+2\sqrt{2})\dPkkk(q))
$$
Hence $NN_{k'}(x)\subset P'$ and $\dPkpkp(x)=\dPkkp(x)\leq\epsilon_k$.
This proves the lemma. 
\end{proof}

\begin{lemma}\label{lem:resamplingOutliers}
Let $P$ be a \uensmp{\epsilon_k}{2} of $K$.
One iteration of decluttering and resampling with parameter $k$ 
and resampling constant $C=10+2\sqrt 2$ 
provides a set $P'\subseteq P$ such that $\delta_H(P',K)\leq 8C\epsilon_k+7\epsilon_k$.
\end{lemma}
\begin{proof}
Let $Q$ denote the output after the decluttering step. Using Theorem~\ref{th:sparsifyingGuarantees} we know that $\delta_H(Q,K)\leq 7\epsilon_k$.
Note that $Q\subset P'$. Thus, we only need to show that for any $p\in P'$, $\dX{p}{K}\leq 8C\epsilon_k+7\epsilon_k$.
Indeed, by the way the algorithm removes points, for any $p \in P'$, there exists $q\in Q$ such that $p\in B(q,C\dPkk(q))$. 
It then follows that 
$$\dX{p}{K}\leq C\dPkk(q) +\dX{q}{K}\leq C(\epsilon_k+\dX{q}{K})+7\epsilon_k\leq 8C\epsilon_k+7\epsilon_k. $$
\end{proof}
 
\begin{lemma}\label{lem:conservationGuaranty}
Given a point $y\in P_i$, there exists $p\in P_0$ such that $\dX{y}{p}\leq \kappa d_{P_i,2^i}(y)$, where $\kappa=\frac{18+17\sqrt{2}}{4}$. 
\end{lemma}
\begin{proof}
We show this lemma by induction on $i$.
First for $i=0$ the claim holds trivially.
Assuming that the result holds for all $j<i$ and taking $y\in P_i$, we distinguish three cases.

\textbf{Case 1:} $y\in P_{i-1}$ and $d_{P_{i-1},2^{i-1}}(y)\leq d_{P_i,2^i}(y)$.\\
Applying the recurrence hypothesis for $j=i-1$ gives the result immediately.

\textbf{Case 2:} $y\notin P_{i-1}$.~~~
It means that $y$ has been removed by decluttering and not been put back by resampling.
These together imply that there exists $q\in Q_i \subseteq P_{i-1}$ such that $\dX{y}{q}\leq 2 d_{P_i,2^i}(y)$ and $\dX{y}{q}>C d_{P_i,2^i}(q)$ with $C = 10+2\sqrt{2}$.
From the proof of Lemma~\ref{lem:resamplingGuarantees}, we know that the $2^{i-1}$ nearest neighbors of $q$ in $P_i$ are resampled and included in $P_{i-1}$. 
Therefore, $d_{P_{i-1},2^{i-1}}(q)=d_{P_i,2^{i-1}}(q)\leq d_{P_i,2^i}(q)$.
Moreover, since $q\in P_{i-1}$, the inductive hypothesis implies that there exists $p\in P_0$ such that $\dX{p}{q}\leq \kappa d_{P_{i-1},2^{i-1}}(q) \leq \kappa d_{P_{i},2^{i}}(q)$. 
Putting everything together, we get that there exists $p\in P_0$ such that
$$\dX{p}{y}\leq\dX{p}{q}+\dX{q}{y} 
\leq\kappa d_{P_i,2^i}(q)+2d_{P_i,2^i}(y)
\leq\left(\frac{\kappa}{5+\sqrt{2}} +2\right)d_{P_i,2^i}(y)
\leq\kappa d_{P_i,2^i}(y).$$
The derivation above also uses the relation that $d_{P_i,2^i}(q) < \frac{1}{C} \dX{y}{q} \le \frac{2}{C} d_{P_i,2^i}(y)$. 

\textbf{Case 3:} $y\in P_{i-1}$ and $d_{P_{i-1},2^{i-1}}(y)> d_{P_i,2^i}(y)$.\\
The second part implies that at least one of the $2^{i-1}$ nearest neighbors of $y$ in $P_i$ does not belong to $P_{i-1}$.
Let $z$ be such a point.
Note that $\dX{y}{z}\leq\sqrt{2} d_{P_i,2^i}(y)$ by Lemma \ref{lem:iNNupperBound}.
For point $z$, we can apply the second case and 
therefore, there exists $p\in P_0$ such that 
\begin{align*}\denselist
\dX{p}{y}& \leq \dX{p}{z}+\dX{z}{y} 
\leq \left(\frac{\kappa}{5+\sqrt{2}} +2\right)d_{P_i,2^i}(z) +\sqrt{2}d_{P_i,2^i}(y)\\
&\leq \left(\frac{\kappa}{5+\sqrt{2}} +2\right)\left(d_{P_i,2^i}(y)+\dX{z}{y}\right) + \sqrt{2}d_{P_i,2^i}(y) \\
&\leq \left(\left(\frac{\kappa}{5+\sqrt{2}} +2\right)(1+\sqrt{2})+\sqrt{2}\right)d_{P_i,2^i}(y)
\leq \kappa d_{P_i,2^i}(y)
\end{align*}\end{proof}

%
\myparagraph{Putting everything together.}  
A repeated application of Lemma~\ref{lem:resamplingGuarantees} (with
weak uniformity) guarantees 
that $P_{i_0+1}$ is a weak \uensmp{\epsilon_{2^{i_0+1}}}{2} of $K$.
One more application (with uniformity) provides that $P_{i_0}$ is a 
\uensmp{\epsilon_{2^{i_0}}}{2} of $K$.
Thus, Lemma~\ref{lem:resamplingOutliers} implies that $d_H(P_{i_0-1},K)\leq  (87+16\sqrt{2})\epsilon_{2^{i_0}}$.
Notice that $P_0\subset P_{i_0-1}$ and thus for any $p\in P_0$, $\dX{p}{K}\leq  (87+16\sqrt{2})\epsilon_{2^{i_0}}$.

To show the other direction, consider any point $x\in K$. 
Since $P_{i_0}$ is a \uensmp{\epsilon_{2^{i_0}}}{2} of $K$, there exists $y\in P_{i_0}$ such that $\dX{x}{y}\leq \epsilon_{2^{i_0}}$ and $d_{P_{i_0},2^{i_0}}(y)\leq2\epsilon_{2^{i_0}}$.
Applying Lemma~\ref{lem:conservationGuaranty}, there exists $p\in P_0$ such that $\dX{y}{p}\leq \frac{18+17\sqrt{2}}{2}\epsilon_{2^{i_0}}$.
Hence $\dX{x}{p}\leq \left(\frac{18+17\sqrt{2}}{2}+1\right)\epsilon_{2^{i_0}}\leq(87+16\sqrt{2})\epsilon_{2^{i_0}}$. 
The theorem then follows. 
\hspace*{\fill}{$\blacksquare$}

\section{Preliminary experimental results}
\label{sec:exp}

We now present some preliminary experimental results for the two denoising algorithms developed in this paper. See Appendix E 
 for more results. 

In Figure \ref{fig:3dmanifoldone}, we show different stages of the {\sf ParfreeDeclutter} algorithm on an input with \emph{adaptively} sampled points. Even though for the parameter-free algorithm, theoretical guarantees are only provided for uniform samples, we note that it performs well on this adaptive case as well.  

\begin{figure}[!ht]
\centering
\includegraphics[height=.4\textwidth]{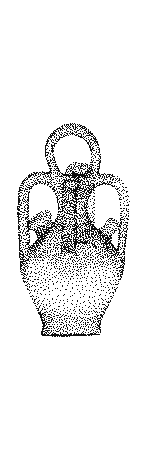}
\includegraphics[height=.4\textwidth]{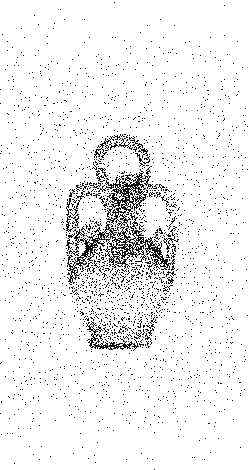}
\includegraphics[height=.4\textwidth]{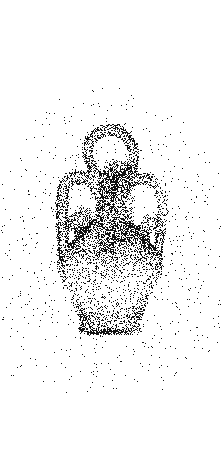}
\includegraphics[height=.4\textwidth]{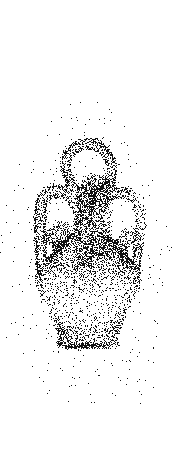}
\includegraphics[height=.4\textwidth]{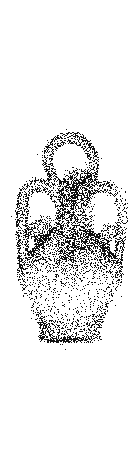}
\vspace*{-0.1in}
\caption{Experiment on a two dimensional manifold in three dimensions.
From left to right, the ground truth, the noisy adaptively sampled input, output of two intermediate steps of Algorithm {\sf ParfreeDeclutter}, and the final result.
}
\label{fig:3dmanifoldone}
\end{figure}

A second example is given in Figure \ref{fig:gps}. 
Here, the input data is obtained from a set of noisy GPS trajectories in the city of Berlin. In particular, given a set of trajectories (each modeled as polygonal curves), we first convert it to a density field by KDE (kernel density estimation). We then take the input as the set of grid points in 2D where every point is associated with a mass (density). Figure \ref{fig:gps} (a) shows the heat-map of the density field where light color indicates high density and blue indicates low density. 
In (b) and (c), we show the output of our {\sf Declutter} algorithm (the {\sf ParfreeDeclutter} algorithm does not provide good results as the input is highly non-uniform) for $k$ = 40 and $k=75$ respectively. 
In (d), we show the set of 40$\%$ points with the highest density values. 
The sampling of the road network is highly non-uniform. In particular, in the middle portion, even points off the roads have very high density due to noisy input trajectories. Hence a simple thresholding cannot remove these points and the output in (d) fills the space between roads in the middle portion; however more aggressive thresholding will cause loss of important roads. 
Our {\sf Declutter} algorithm can capture the main road structures without collapsing nearby roads in the middle portion though it also sparsifies the data. 

\begin{figure}
\begin{tabular}{cccc}
        \includegraphics[width=3cm]{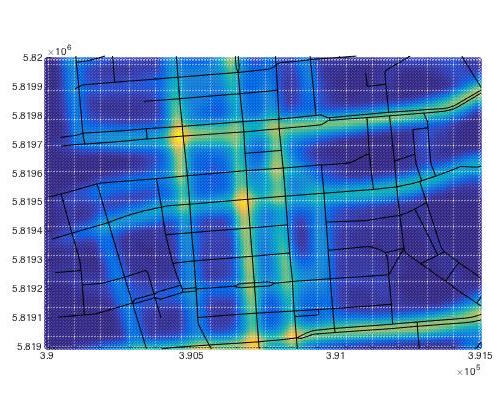} &
        \includegraphics[width=3cm]{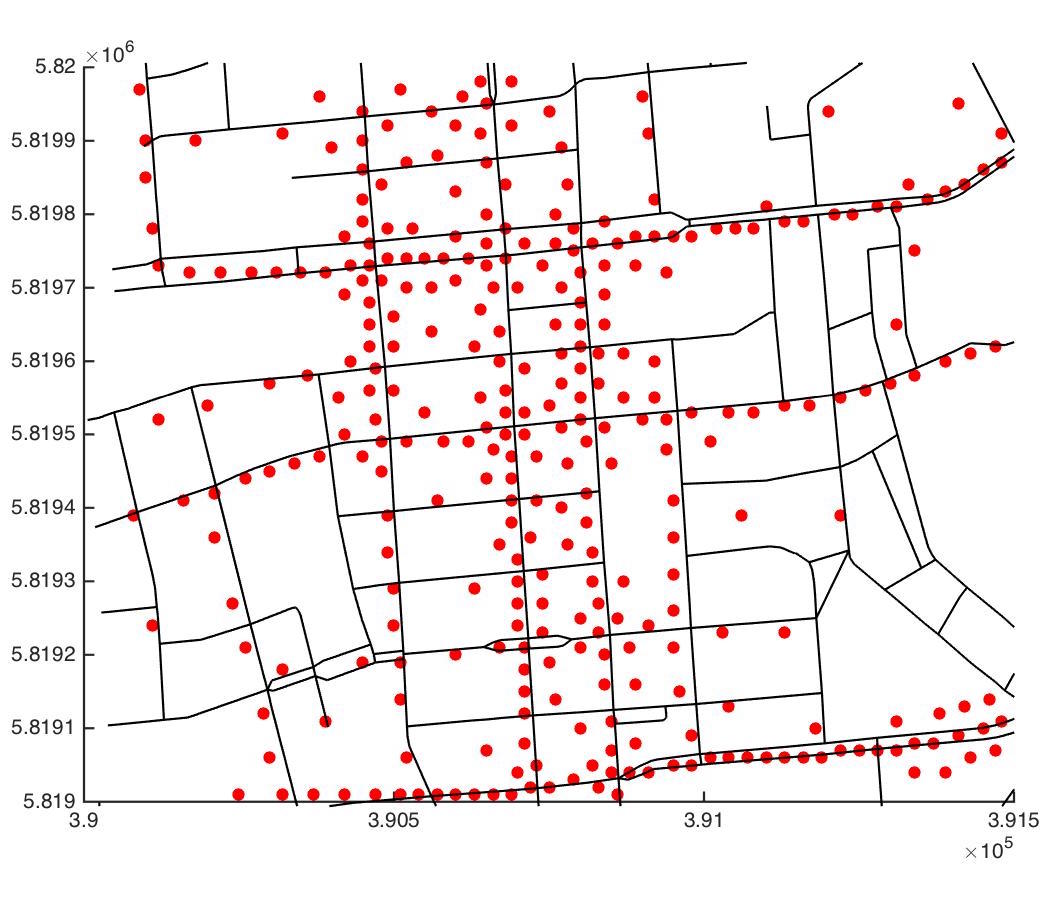} &
        \includegraphics[width=3cm]{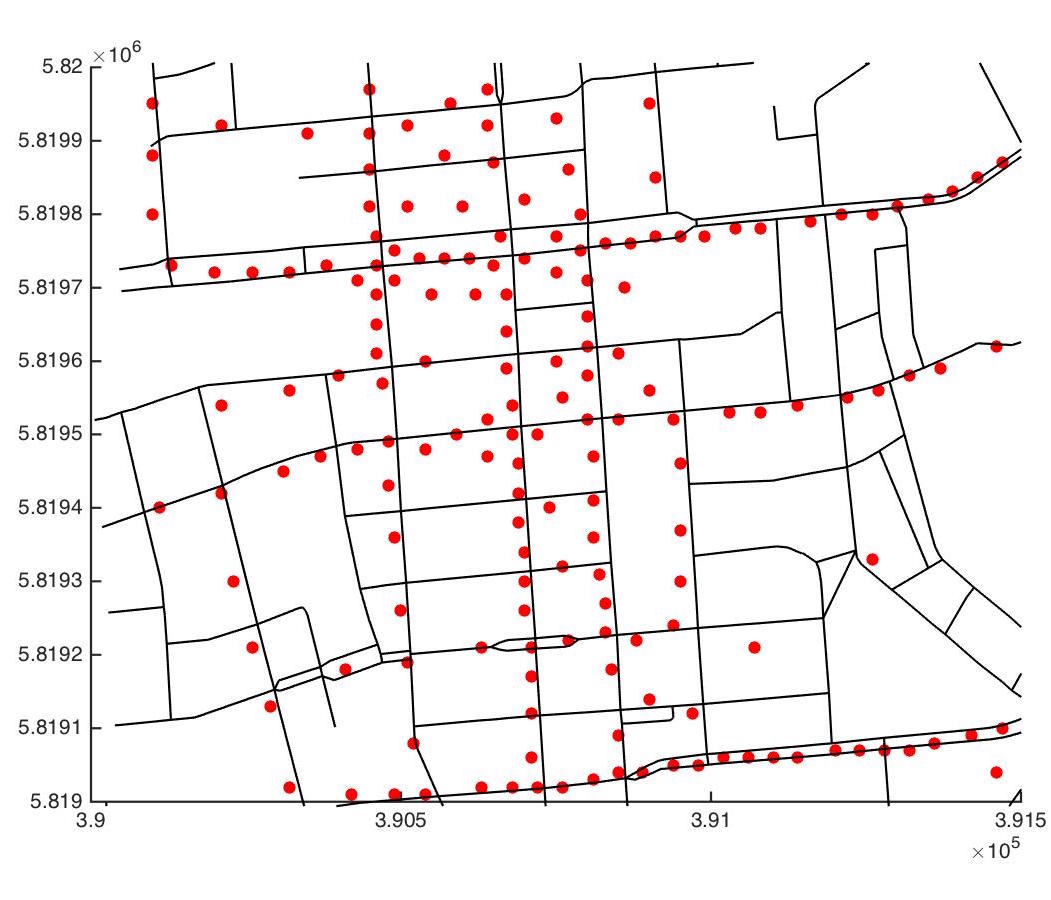} &
        \includegraphics[width=3cm]{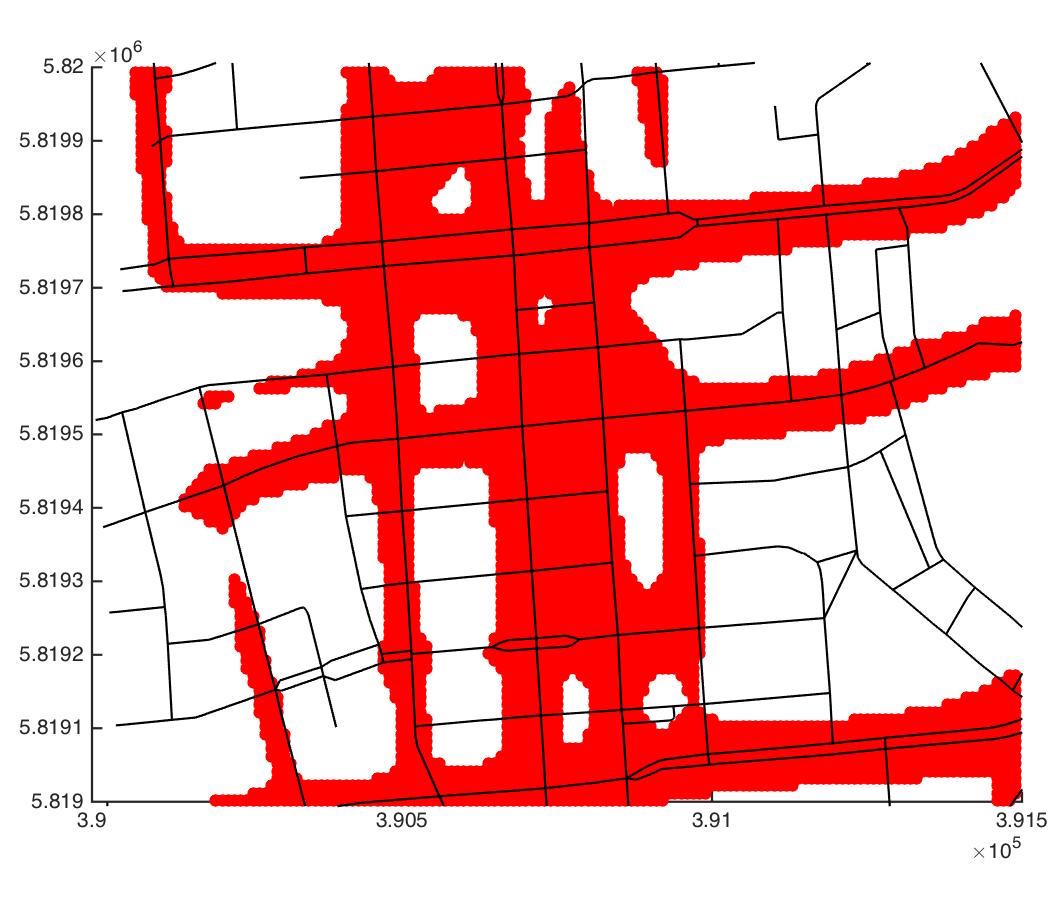} \\
(a) & (b) & (c) & (d)
\end{tabular}
\vspace*{-0.15in}
    \caption{(a) The heat-map of a density field generated from GPS traces. There are around 15k (weighted) grid points serving as an input point set. The output of Algorithm {\sf Declutter} when (b) $k=40$ and (c) $k=75$, (d) thresholding of 40\% points with the highest density.}
    \label{fig:gps}
\vspace*{-0.15in}
\end{figure}

In another experiment, we apply the denoising algorithm as a pre-processing for high-dimensional data classification. Here we use MNIST data sets, which is a database of handwritten digits from '0' to '9'. 
Table \ref{tab:17cls} shows the experiment on digit 1 and digit 7. We take a random collection of 1352 images of digit '1' and 1279 images of digit '7' correctly labeled as a training set, and take 10816 images of digit 1 and digit 7 as a testing set. Each of the image is $28 \times 28$ pixels and thus can be viewed as a vector in $\mathbb{R}^{784}$. We use the $L_1$ metric to measure distance between such image-vectors. We use a linear SVM to classify the 10816 testing images. The classification error rate for the testing set is $0.6564\%$  shown in the second row of Table \ref{tab:17cls}. 

Next, we artificially add two types of noises to input data: the {\it swapping-noise} and the {\it background-noise}. 
The swapping-noise means that we randomly mislabel some images of `1' as '7', and some images of `7' as '1'. As shown in the third row of Table \ref{tab:17cls}, the classification error increases to about $4.096\%$ after such mislabeling in the training set. 

Next, we apply our {\sf ParfreeDeclutter} algorithm to this training set with added swapping-noise (to the set of images with label '1' and the set with label '7' separately) to first clean up the training set. As we can see in Row-6 of Table \ref{tab:17cls}, we removed most images with a mislabeled `1' {(which means the image is '7' but it is labeled as '1')}. A discussion on why mislabeled `7's are not removed is given in Appendix E. We then use the denoised dataset as the new training set, and improved the classification error to $2.45\%$. 

The second type of noise is the {\it background noise}, where we replace the black backgrounds of a random subset of images in the training set (250 `1's and 250 `7's) with some other grey-scaled images. Under such noise, the classification error increases to $1.146\%$. Again, we perform our {\sf ParfreeDeclutter} algorithm to denoise the training sets, and use the denoised data sets as the new training set. The classification error is then improved to $0.7488\%$. 
More results on the MNIST data sets are reported in Appendix E. 

%
\begin{center}
\begin{table}
\centering
\resizebox{\columnwidth}{!}{\hspace*{\fill}

    \begin{tabular}{|c |c | c | c | c | c | p{1.5cm} |}
    \hline
    \rownumber&\multicolumn{5}{c|}{}&Error(\%) \\ \hline
   \rownumber& Original & \multicolumn{2}{c|}{\# Digit 1\quad 1352} &   \multicolumn{2}{c|}{\# Digit 7\quad 1279}  &0.6564 \\ \hline 
 \multicolumn{7}{c}{}
   \\ \hline
      \rownumber&Swap. Noise& \multicolumn{2}{c|}{\# Mislabelled 1\quad 270} & \multicolumn{2}{c|}{\# Mislabelled 7\quad 266} &4.0957 \\ \hline

   \rownumber&&\multicolumn{2}{c|}{Digit 1}&\multicolumn{2}{c|}{Digit 7}& \\ \hline
    \rownumber&&\# Removed&\# True Noise&\# Removed&\# True Noise &\\ \hline
    \rownumber& L1 Denoising&314&264&17&1&2.4500\\  \hline

  \end{tabular}
  }
         \bigbreak
        \resizebox{\columnwidth}{!}{
        \begin{tabular}{|c| c | c | c | c | c | p{1.5cm} |}
        \hline
    \rownumber&Back. Noise& \multicolumn{2}{c|}{\# Noisy 1\quad 250} &\multicolumn{2}{c|}{\# Noisy 7\quad 250}  &1.1464 \\ \hline
    \rownumber& &\multicolumn{2}{c|}{Digit 1}&\multicolumn{2}{c|}{Digit 7}& \\ \hline
    \rownumber&&\# Removed&\# True Noise&\# Removed&\# True Noise &\\ \hline
     \rownumber&L1 Denoising&294&250&277&250&0.7488\\ \hline
    \end{tabular}
    }
    \caption{Results of denoising on digit 1 and digit 7 from the MNIST. }
    \label{tab:17cls}
  \end{table} 
\vspace*{-0.2in}
\end{center}

\section{Discussions}
\label{sec:conclusion}

Parameter selection is a notorious problem for many algorithms in practice. 
Our high level goal is to understand the roles of parameters in algorithms for denoising, how to reduce their use and what theoretical guarantees do they entail. 
While this paper presented some results towards this direction, many interesting questions ensue. 
For example, how can we further relax our sampling conditions, making them allow more general inputs, and how to connect them with other classical noise models? 

We also note that while the output of {\sf ParfreeDeclutter}
is guaranteed to be close to the ground truth w.r.t. the Hausdorff distance, this Hausdorff distance itself is not estimated. 
Estimating this distance appears to be difficult. 
We could estimate it if we knew the correct scale, i.e. $i_0$, 
to remove the ambiguity. 
Interestingly, even with the uniformity condition, 
it is not clear how to estimate this distance in a parameter free manner. 

We do not provide guarantees for the parameter-free algorithm in an adaptive setting though the algorithm behaved well empirically for the adaptive case too. 
A partial result is presented in Appendix B 
,
but the need for a small $\epsilon_k$ in the conditions defeat the attempts to obtain a complete result. 

The problem of parameter-free denoising under more general sampling conditions remains
open. 
It may be possible to obtain results by replacing uniformity with other assumptions, for example topological assumptions: say, if the ground truth is a simply connected manifold without boundaries, can this help to denoise and eventually reconstruct the manifold?

\subparagraph*{Acknowledgments}

We thank Ken Clarkson for pointing out the result in~\cite{chan2006spanners}. 
This work is in part supported by National Science Foundation via grants CCF-1618247, CCF-1526513, CCF-1318595 and IIS-1550757. 

\bibliography{parameter_free_full}
\newpage
\appendix

\section{Missing details from section \ref{sec:decluttering}}
\label{appendix:sec:decluttering}
\paragraph*{Proof of Lemma \ref{lem:adaptiveDensity}.}


Let $x$ be a point of $K$.
Then there exists $i$ such that $\dX{p_i}{x}\leq \dPk(x)\leq \epsilon_k f(x)$.
If $p_i$ belongs to $\outputP$, then setting $q=p_i$ proves the lemma.
Otherwise, because of the way that the algorithm eliminates points, there must exist $j<i$ such that $p_j\in \outputP_{i-1}\subseteq \outputP$ and
$$\dX{p_i}{p_j}\leq 2\dPkk(p_i)\leq 2(\dPk(x)+\dX{p_i}{x})\leq 4\epsilon_k f(x),$$
the second inequality follows from the 1-Lipschitz property of $\dPkk$ function and the sampling Condition 1. Then
$$\dX{x}{p_j}\leq\dX{x}{p_i}+\dX{p_i}{p_j}\leq 5\epsilon_kf(x).$$ 

\paragraph*{Proof of Lemma \ref{lem:closeness}.}
Consider any $p_i\in P$ and let $\bar p_i$ be one of its nearest
points in $K$.
It is sufficient to show that if
$\dX{p_i}{\bar p_i}>7\epsilon_kf(\bar p_i)$, then $p_i\notin \outputP$ . 

By Condition 2 of Def. \ref{def:adasample}, $\dPkk(p_i)\geq \dX{p_i}{\bar p_i}-\epsilon_kf(\bar p_i)>6\epsilon_kf(\bar p_i)$.
By Lemma~\ref{lem:adaptiveDensity}, there exists $q\in \outputP$ such that $\dX{\bar p_i}{q}\leq 5\epsilon_kf(\bar p_i)$. Thus,
$$
\dPkk(q)\leq \dPkk(\bar{p}_i)+\dX{\bar p_i}{q} \leq 6\epsilon_kf(\bar p_i).
$$
Therefore, $\dPkk(p_i)> 6\epsilon_kf(\bar p_i)\geq \dPkk(q)$ implying that $q\in Q_{i-1}$.
Combining triangle inequality and Condition 2 of Def. \ref{def:adasample}, we have 
$$\dX{p_i}{q}\leq \dX{p_i}{\bar p_i}+\dX{\bar p_i}{q}
\leq \dPkk(p_i)+\epsilon_kf(\bar p_i)+5\epsilon_kf(\bar p_i)
< 2\dPkk(p_i). $$
Therefore, $q\in Q_{i-1}\cap B(p_i,2\dPkk(p_i))$, meaning that $p_i\notin \outputP$. 

Hence, we have a point of $Q_{i-1}$ inside the ball of center $p_i$ and radius $2\dPkk(p_i)$, which guarantees that $p_i$ is not selected.
The lemma then follows.

\section{Towards parameter-free denoising for adaptive case}
\label{appendix:sec:parafree}

Unfortunately, our parameter-free denoising algorithm does not fully work in the adaptive setting.
We can still prove  that one iteration of the loop works. 
However, the value chosen for the resampling
constant $C$ has to be sufficiently large with respect to $\epsilon_k$.
This condition is not satisfied
when $k$ is large as $\epsilon_k$ in that case is very large.

\begin{theorem}
Let $P$ be a point set that is both a \luensmp{\epsilon_{2k}}{2} and a \luensmp{\epsilon_{k}}{2} of $K$. 
Applying one step of the {\sf ParfreeDeclutter} algorithm with parameter $2k$ gives a point set $P'$ which is a \luensmp{\epsilon_{2k}}{2} of $K$ when $\epsilon_{2k}$ is sufficiently small and the resampling constant $C$ is sufficiently large.
\end{theorem}

\begin{proof}
As in the global conditions case, only the first condition has to be checked.
Let $x\in K$ then, following the proof of Lemma~\ref{lem:adaptiveDensity}, there exists $q\in P'$ such that $\dX{x}{q}\leq5\epsilon_{2k}$ and $\dPkkk(q)\leq 2\epsilon_{2k}$.
The feature size $f$ is 1-Lipschitz and thus:
\begin{align*}
f(x)&\leq f(\bar q)+\dX{\bar q}{x}\\
&\leq f(\bar q)+\dX{q}{\bar q}+\dX{q}{x}\\
&\leq f(\bar q)+\dPkkk(q)+\epsilon_{2k} f(\bar q)+5\epsilon_{2k} f(x)
\end{align*}
Hence $$f(\bar q)\geq \frac{1-7\epsilon_{2k}}{1+\epsilon_{2k}} f(x).$$
Therefore $\dPkkk(q)\geq \frac{1-7\epsilon_{2k}}{1+\epsilon_{2k}}\frac{\epsilon_{2k}}{2} f(x)$.
The claimed result is obtained if the constant $C$ satisfies $C\geq\frac{2(5+\sqrt{2})(1+\epsilon_{2k})}{1-7\epsilon_{2k}}$ as $B(x,\sqrt{2}\epsilon_{2k}f(x))\subset B(q,C\dPkkk(q))$.
\end{proof}

\section{Application to topological data analysis}\label{sec:tda}

In this section, we provide an example of using our decluttering algorithm for topology inference. 
We quickly introduce notations for some notions of algebraic topology and refer the reader to~\cite{ctaiEH,atH,eatM} for the definitions and basic properties. 
Our approaches mostly use standard arguments from the literature of topology inference; e.g, \cite{MEscalar,tpbresCO,sdlutimDDW}. 

Given a topological space $X$, we denote $H_i(X)$ its $i$-dimensional
homology group with coefficients in a field.
As all our results are independent of $i$, we will write $H_*(X)$.
We consider the persistent homology of filtrations obtained as sub-level sets of distance functions.
Given a compact set $K$, we denote the distance function to $K$ by $d_K$.
We moreover assume that the ambient space is triangulable which ensures that these functions are tame and the persistence diagram $\Dgm{d_K^{-1}}$ is well defined.
We use $d_B$ for the bottleneck distance between two persistence diagrams.
We recall the main theorem from~\cite{spdCEH} which implies:

\begin{proposition}
Let $A$ and $B$ be two triangulable compact sets in a metric space. Then,
$$d_B(\Dgm{d_A^{-1}},\Dgm{d_B^{-1}})\leq d_H(A,B).$$
\end{proposition}

This result trivially guarantees that the result of our decluttering algorithm allows us to approximate the persistence diagram of the ground truth.

\begin{corollary}
Given a point set $P$ which is an \ensmp{\epsilon_k} of a compact set $K$, the {\sf Declutter} algorithm returns a set $\outputP$ such that
$$d_B(\Dgm{d_K^{-1}},\Dgm{d_{\outputP}^{-1}})\leq 7\epsilon_k.$$
\end{corollary}

The algorithm reduces the size of the set needed to compute an approximation diagram.
Previous approaches relying on the distance to a measure to handle noise ended up with a weighted set of size roughly $n^k$ or used multiplicative approximations which in turn implied a stability result at logarithmic scale for the Bottleneck distance~\cite{MEsparse,wkdGMM}.
The present result uses an unweighted distance to compute the persistence diagram and provides guarantees without the logarithmic scale using fewer points than before.

If one is interested in inferring homology instead of computing a persistence diagram, our previous results guarantee that the \v Cech complex $C_\alpha(\outputP)$ or the Rips complex $R_\alpha(\outputP)$ can be used.
Following~\cite{tpbresCO}, we use a nested pair of filtration to remove noise.
Given $A\subset B$, we consider the map $\phi$ induced at the homology level by the inclusion $A\hookrightarrow B$. We denote $H_*(A\hookrightarrow B)=\mathrm{Im}(\phi)$.
More precisely, denoting $K^\lambda=d_K^{-1}(\lambda)$ and $\wfs$ as
the weak feature size, we obtain:

\begin{proposition}\label{prop:CechGuarantees}
Let $P$ be an \ensmp{\epsilon_k} of a 
compact set $K\subset\R^d$ with $\epsilon_k<\frac{1}{28}\wfs(K)$.
Let $\outputP$ be the output of {\sf Declutter}($P$).
Then for all $\alpha$, $\alpha'\in[7\epsilon_k,\wfs(K)-7\epsilon_k]$ such that $\alpha'-\alpha>14\epsilon_k$ and for all $\lambda\in(0,\wfs(K))$, we have
$$H_*(K^\lambda)\cong H_*(C_\alpha(\outputP)\hookrightarrow C_{\alpha'}(\outputP))$$
\end{proposition}

\begin{proposition}
Let $P$ be an \ensmp{\epsilon_k} of a compact 
set $K\subset\R^d$ with $\epsilon_k<\frac{1}{35}\wfs(K)$.
Let $\outputP$ be the output of {\sf Declutter}($P$).
Then for all $\alpha\in[7\epsilon_k,\frac{1}{4}(\wfs(K)-7\epsilon_k)]$ and $\lambda\in(0,\wfs(K))$, we have
$$H_*(K^\lambda)\cong H_*(R_\alpha(\outputP)\hookrightarrow R_{4\alpha}(\outputP))$$
\end{proposition}

These two propositions are direct consequences of~\cite[Theorems 3.5 \& 3.6]{tpbresCO}.
To be used, both these results need the input of one or more parameters, $\alpha$ and $\alpha'$, corresponding to a choice of scale.
This cannot be avoided as it is equivalent to estimating the Hausdorff distance between a point set and an unknown compact set, problem discussed in the introduction.
However, by adding a uniformity hypothesis and knowing the uniformity constant $c$, the problem can be solved.
We use the fact that the minimum $\dPkk$ over the point set $P$ is bounded
from below.
Let us write $\kappa=\min_{p\in P} \dPkk(p)$.

\begin{lemma}
If $P$ is an \ensmp{\epsilon_k} of $K$ then $\kappa\leq 2\epsilon_k$.
\end{lemma}

\begin{proof}
Let $x\in K$, then there exists $p\in P$ such that $\dX{x}{p}\leq\dPkk(x)\leq\epsilon_k$.
Therefore $\kappa\leq\dPkk(p)\leq\dPkk(x)+\dX{x}{p}\leq 2\epsilon_k$.
\end{proof}

This trivial observation has the consequence that $c$ is greater than $\frac{1}{2}$ in any \uensmp{\epsilon_k}{c}.
We can compute $c\kappa$ and use it to define an $\alpha$ for using the previous propositions.
We formulate the conditions precisely in the following propositions.
Note that the upper bound for $\alpha$ is not necessarily known.
However, the conditions imply that the interval of correct values for $\alpha$ is non-empty.

\begin{proposition}
Let $P$ be a \uensmp{\epsilon_k}{c} of a compact 
set $K\subset\R^d$ with $c\epsilon_k<\frac{1}{56}\wfs(K)$.
Let $\outputP$ be the output of {\sf Declutter}($P$).
Then for all $\alpha$, $\alpha'\in[7c\kappa,\wfs(K)-7c\epsilon_k]$ such that $\alpha'-\alpha>14c\kappa$ and for all $\lambda\in(0,\wfs(K))$, we have
$$H_*(K^\lambda)\cong H_*(C_\alpha(\outputP)\hookrightarrow C_{\alpha'}(\outputP))$$
\end{proposition}

\begin{proof}
Following Proposition~\ref{prop:CechGuarantees}, we need to choose $\alpha$ and $\alpha'$ inside the interval $[7\epsilon_k,\wfs(K)-7\epsilon_k]$.
Using the third hypothesis, we know that $7c\kappa\geq7c\epsilon_k$.
We need to show that $\alpha$ and $\alpha'$ exist, i.e. $21 c\kappa<\wfs(K)-7\epsilon_k$.
Recall that $c\geq 2$ , $\kappa\leq 2\epsilon_k$. Therefore,
$21c\kappa+7\epsilon_k\leq 56 c\epsilon_k<\wfs(K)$.
\end{proof}

\begin{proposition}
Let $P$ be a \uensmp{\epsilon_k}{c} of a compact 
set $K\subset\R^d$ with $c\epsilon_k<\frac{1}{70}\wfs(K)$.
Let $\outputP$ be the output of {\sf Declutter}($P$).
Then for all $\alpha\in[7c\kappa,\frac{1}{4}(\wfs(K)-7\epsilon_k)]$ and $\lambda\in(0,\wfs(K))$, we have
$$H_*(K^\lambda)\cong H_*(R_\alpha(Q)\hookrightarrow R_{4\alpha}(Q))$$
\end{proposition}

The proof is similar to the one for the previous proposition.
Note that even if the theoretical bound can be larger, we can always pick $\alpha=7c\kappa$ in the second case and the proof works.
The sampling conditions on these results can be weakened by using the more general notion of $(\epsilon_k,r,c)$-sample of \cite{MEthesis}, assuming that $r$ is sufficiently large with respect to $\epsilon_k$.

\newcommand{\cmone}		{{c_\X}}
\newcommand{\cmtwo}		{{c_{Lip}}}

\section{Extensions for Declutter algorithm}
\label{appendix:relax}

It turns out that our {\sf Declutter} algorithm can be run with different choices for the $k$-distance $\dPkk(x)$ as introduced in Definition \ref{def:kdist} which still yields similar denosing guarantees. 

Specifically, assume that we now have a certain robust distance estimate $\dPkk(x)$ for each point $x\in \X$ such that the following properties are satisfied. 

\begin{description}
\item[Conditions for $\dPkk$]. \\
(A) For any $x \in \X$, $d_\X (x, P) \le \dPkk(x)$; and \\
(B) $\dPkk$ is $1$-Lipschitz, that is, for any $x, y \in \X$ we have $\dPkk(x) \le \dPkk(y) + d(x, y)$. 
\end{description}

\begin{description}
\item[Examples]. \\
(1) We can set $\dPkk$ to be the average distance to $k$ nearest neighbors in $P$; that is, for any $x\in \X$, define $\dPkk(x) = \frac{1}{k}\sum_{i=1}^k d(x, p_i(x))$ where $p_i(x)$ is the $i$th nearest neighbor of $x$ in $P$. 
We refer to this as the \emph{average $k$-distance}. It is easy to show that the average $k$-distance satisfies the two conditions above. \\
(2) We can set $\dPkk(x)$ to be the distance from $x$ to its $k$-th nearest neighbors in $P$; that is, $\dPkk(x) = d_\X (x, p_k(x))$. 
We refer to this distance as {\it $k$-th NN-distance}. 
\end{description}

We can then define the sampling condition (as in Definitions \ref{def:sampling} and \ref{def:unisample}) based on our choice of $\dPkk$ as before. Notice that, under different choices of $\dPkk$, $P$ will be an $\varepsilon_k$-noisy sample for different values of $\varepsilon_k$.
Following the same arguments as in Section \ref{sec:decluttering}, we can show that Theorems \ref{th:sparsifyingGuarantees} and \ref{thm:adaptiveHaus} still hold as long as $\dPkk$ satisfies the two conditions above. For clarity, we provide an explicit statement for the analog of Theorem \ref{th:sparsifyingGuarantees} below and omit the corresponding statement for Theorem \ref{thm:adaptiveHaus}. 

 \begin{theorem}\label{thm:generalDeclutter}
Given a \ensmp{\epsilon_k} $P$ of a compact 
set $K\subseteq\X$ under a choice of $\dPkk$ that satisfies the conditions (A) and (B) as stated above, Algorithm 
{\sf Declutter} returns a set $\outputP\subseteq P$ such that $$d_H(K,\outputP)\leq 7\epsilon_k.$$
\end{theorem}

We remark that in \cite{chan2006spanners}, Chan et al. proposed to use the $k$-th NN-distance to generate the so-called $\varepsilon$-density net, where $k = \varepsilon n$. The criterion to remove points from $P$ to generate $Q$ in their procedure is slightly different from our 
{\sf Declutter} algorithm. However, it is easy to show that the output of their procedure (which is a $k/n$-density net) satisfies the same guarantee as the output of the {\sf Declutter} algorithm does (Theorem \ref{thm:generalDeclutter}). 

\myparagraph{Further extensions.}
One can in fact further relax the conditions on $\dPkk(x)$ or even on the input metric space $(\X, d_\X)$ such that the triangle inequality for $d_\X$ only approximately holds. In particular, we assume that 
\begin{description}
\item[Relaxation-1]. $d_\X(x,y) \le \cmone [d_\X(x, w) + d_\X(w, y)]$, for any $x, y, w \in \X$, with $\cmone \ge 1$. That is, the input ambient space $(X, d_\X)$ is almost a metric space where the triangle inequality holds with a multiplicative factor. 
\item[Relaxation-2]. $\dPkk(x) \le \cmtwo[\dPkk(y) + d_\X(x, y)]$, for any $x, y \in \X$ with $\cmtwo \ge 1$. That is, the $1$-Lipschitz condition on $\dPkk$ is also relaxed to have a multiplicative factor. 
\end{description}
We then obtain the following analog of Theorem \ref{th:sparsifyingGuarantees}: 
\begin{theorem}\label{thm:relaxDeclutter}
Suppose $(\X, d_\X)$ is a space where $d_\X$ satisfies {\sf Relaxation-1} above. 
Let $\dPkk$ be a robust distance function w.r.t. $P$ that satisfies {\sf Condition (A)} and {\sf Relaxation-2} defined above. 
Given a point set $P$ which is an \ensmp{\epsilon_k} of a compact 
set $K\subseteq\X$ under the choice of $\dPkk$, Algorithm 
{\sf Declutter} returns a set $\outputP\subseteq P$ such that
$$d_H(K,Q)\leq m\epsilon_k$$
where $m=max\left\{\cmtwo+\cmone\cmtwo+4\cmone\cmtwo^2+1,\frac{2+\cmone^2+4\cmone^2\cmtwo}{2-\cmone}\right\}.$
\end{theorem}

Finally, we remark that a different way to generalize the 1-Lipschitz condition for $\dPkk(x)$ is by asserting $\dPkk(x) - \dPkk(y) \le \cmtwo d_\X(x,y)$. We can use this to replace  {\sf Relaxation-2} and obtain a similar guarantee as in the above theorem. 
We can also further generalize {\sf Relaxation-2} by allowing an additive term as well. We omit the resulting bound on the output of the {\sf Declutter} algorithm.

\section{Experimental results}
\label{appendix:sec:exp}

In this section, we provide more details 
of our empirical results, an abridged version of which
already appeared in the main text. 
We start with the decluterring algorithm.
This algorithm needs the input of a parameter $k$.
This parameter has a direct influence on the result.
On one hand, if $k$ is too small, not all noisy points are removed from the sample.
On the other hand, if $k$ is too large, we remove too many points and end up with a very sparse sample that is unable to describe the underlying object precisely.

\paragraph*{Experiments for {\sf Declutter} algorithm.} 
Figure~\ref{fig:declutterfail} presents results of Algorithm {\sf Declutter}  for the so-called Botijo example.
In this case, no satisfying $k$ can be found.
A parameter $k$ that is sufficiently large to remove the noise creates an output set that is too sparse to describe the ground truth well.

\begin{figure}[!ht]
\centering
\includegraphics[width=.24\textwidth]{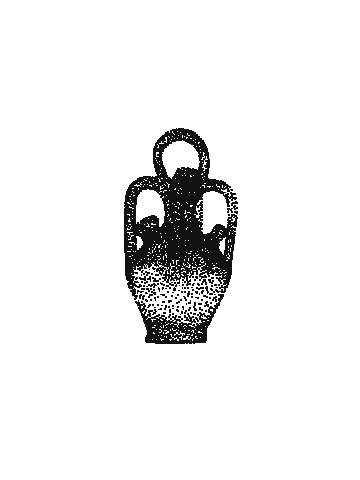}
\includegraphics[width=.24\textwidth]{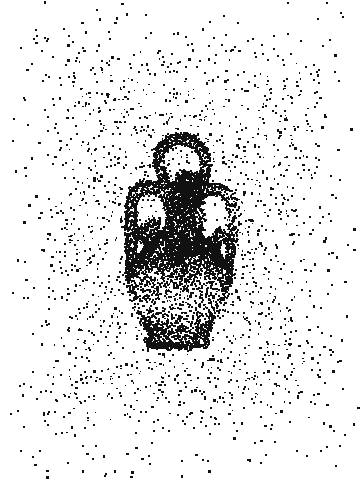}
\includegraphics[width=.24\textwidth]{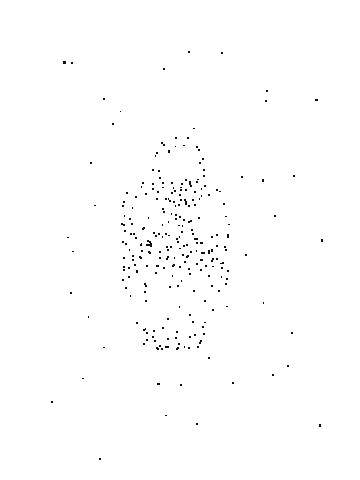}
\includegraphics[width=.24\textwidth]{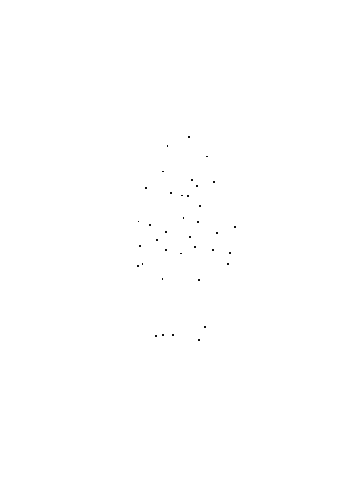}
\caption{From left to right, the ground truth, the noisy input and the output of Algorithm {\sf Declutter} for $k=81$ and $k=148$}\label{fig:declutterfail}
\end{figure}

We further illustrate the behavior of our algorithm by looking at the Hausdorff distance between the output and the ground truth, and at the cardinality of the output, in the function of $k$ (Figure~\ref{fig:decluttergraph}).
Note that the Hausdorff distance drops suddenly when we remove the last of the outliers.
However, it is already too late to represent the ground truth well
as only a handful of points are kept at this stage.
While sparsity is often a desired property, here it becomes a hindrance as we are no longer able to describe the underlying set.

\begin{figure}[!ht]
\centering
\includegraphics[height=.35\textwidth]{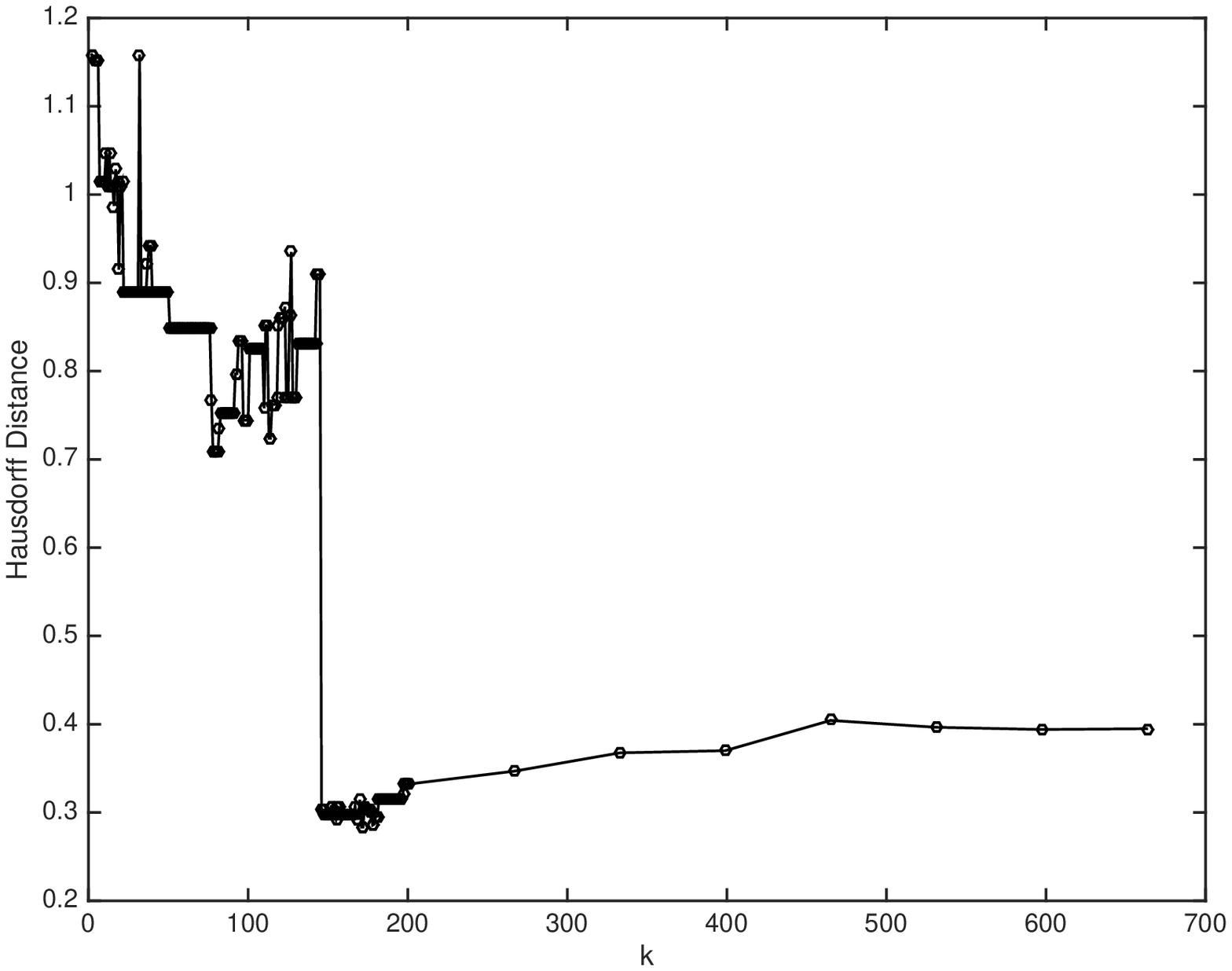}
\includegraphics[height=.35\textwidth]{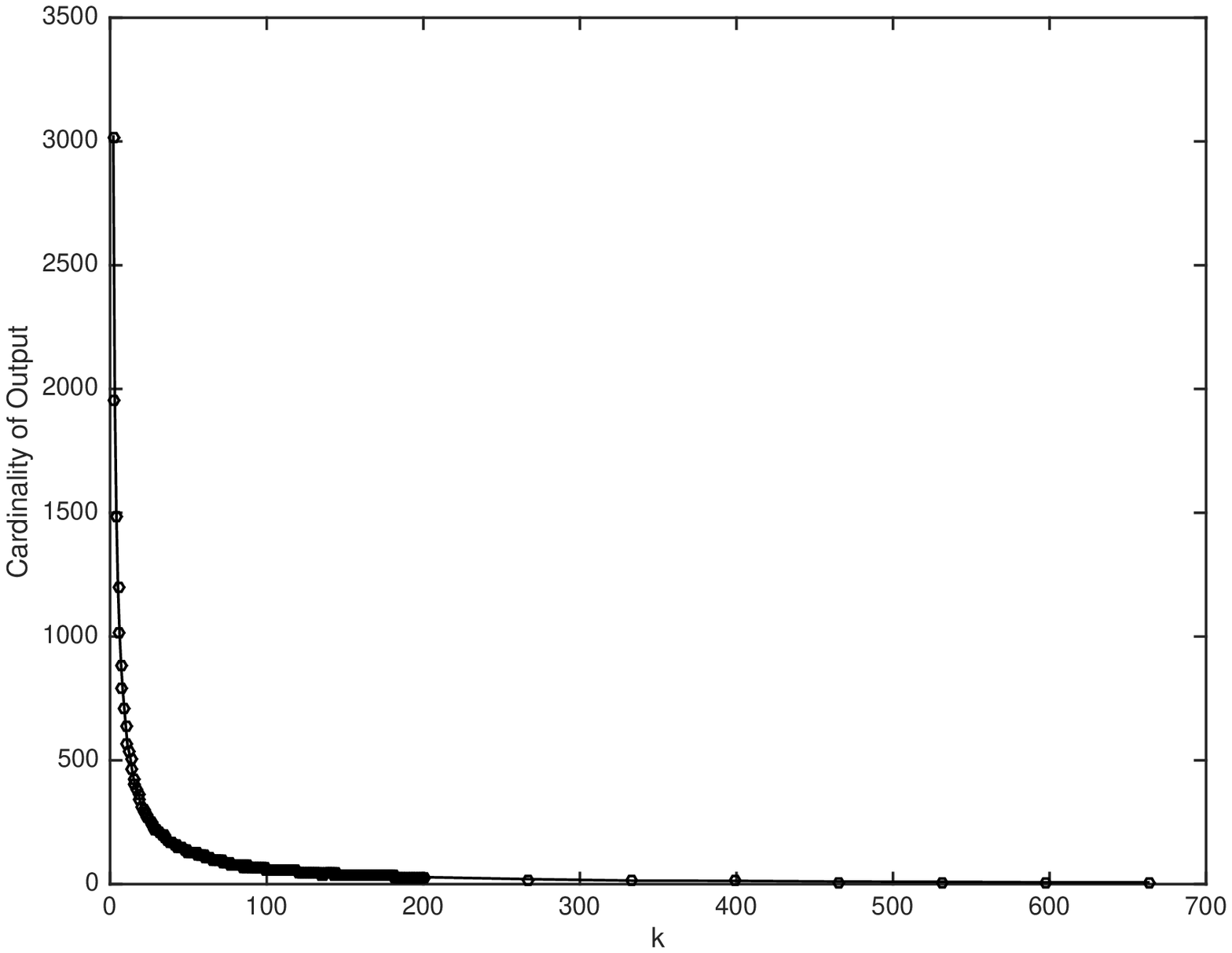}
\caption{Hausdorff distance between the ground truth and the output of the declutter algorithm, and cardinality of this output in the function of $k$.}
\label{fig:decluttergraph}
\end{figure}

The introduction of the resample step allows us to solve this sparsity problem.
If we were able to choose the right parameter $k$, we could simply sparsify and then resample to get a good output.
One can hope that the huge drop in the left graph could be used to choose the parameter.
However, the knowledge of the ground truth is needed to compute it, and estimating the Hausdorff distance between a set and the ground truth is impossible without some additional assumptions like the uniformity we use.

A second example is given in Figure \ref{fig:gps-two}. 
Here, the input data is obtained from a set of noisy GPS trajectories in the city of Berlin. In particular, given a set of trajectories (each modeled as polygonal curves), we first convert it to a density field by KDE (kernel density estimation). We then take the input as the set of grid points in 2D where every point is associated with a mass (density). Figure \ref{fig:gps-two} (a) shows the heat-map of the density field where light color indicates high density and blue indicates low density. 
In (b) and (c), we show the outputs of  {\sf Declutter} algorithm (the {\sf ParfreeDeclutter} algorithm would not provide good results as the input is highly non-uniform) for $k$ = 40 and $k=75$ respectively. 
In (d), we show the set of 40$\%$ points with the highest density values. 
The sampling of the road network is highly non-uniform. In particular, in the middle portion, even points off the roads have very high density (due to noisy input trajectories) as well. Hence a simple thresholding cannot remove these points and the output in (d) fills the space between roads in the middle portion. If we increase the threshold further, that is, reduce the number of points we want to keep in the thresholding, we will lose structures of some main roads. Our {\sf Declutter} algorithm can capture the main road structures without collapsing nearby roads in the middle portion, although it also sparsifies the data. 

\begin{figure}
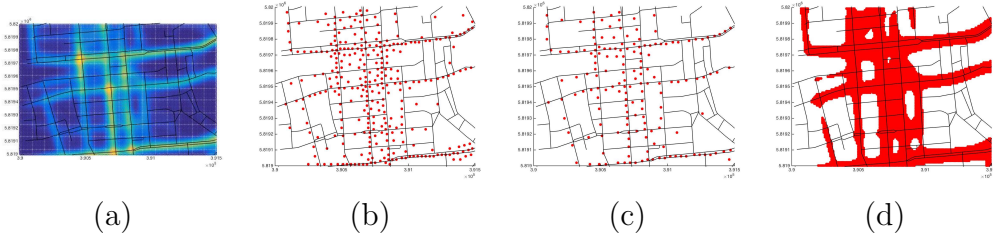

\begin{tabular}{cccc}
        \includegraphics[width=3cm]{15k_gt} &
        \includegraphics[width=3cm]{k40_316} &
        \includegraphics[width=3cm]{k75_168} &
        \includegraphics[width=3cm]{15k_04thresh} \\
(a) & (b) & (c) & (d)
\end{tabular}
    \caption{(a)The heat-map of a density field generated from GPS traces. There are around 15k (weighted) grid points, serving as an input point set. The output of Algorithm {\sf Declutter} when (b) $k=40$ and (c) $k=75$, (d) Thresholding of 40\% points with the highest density.}
    \label{fig:gps-two}
\end{figure}

\paragraph*{Experiments on {\sf ParfreeDeclutter} algorithm.} 
We will now illustrate our parameter-free denoising algorithm on several examples. 
Recall that the theoretical guarantee of the output of our parameter-free algorithm (i.e, {\sf ParfreeDeclutter} algorithm) so far is only provided for samples satisfying some uniformity conditions. 

We start with some curves in the plane.
Figure~\ref{fig:1d} shows the results on two different inputs.
In both cases, the curves have self-intersections.
The noisy inputs are again obtained by moving every input point according to a Gaussian distribution and adding some white background noise.
The details of the noise models can be found in Table~\ref{tab:noise} 
and the details on the size of the various point sets are given in Table~\ref{tab:cardinality}.

The first steps of the algorithm remove the outliers lying further away from the ground truth.
As the value of the parameter $k$ decreases, we remove nearby outliers.
The result is a set of points located around the curves, in a tubular neighborhood of width that depends on the standard deviation of the Gaussian noise.
Small sharp features are lost due to the blurring created by the Gaussian noise but the Hausdorff distance between the final output and the ground truth is as good as one can hope for when using a method oblivious of the ground truth.

\begin{figure}[!ht]
\centering
\includegraphics[height=.24\textwidth]{groundtruth}
\includegraphics[height=.24\textwidth]{noise}
\includegraphics[height=.24\textwidth]{resamplek128}
\includegraphics[height=.24\textwidth]{resamplek64}
\includegraphics[height=.24\textwidth]{resamplek1}
\includegraphics[height=.24\textwidth]{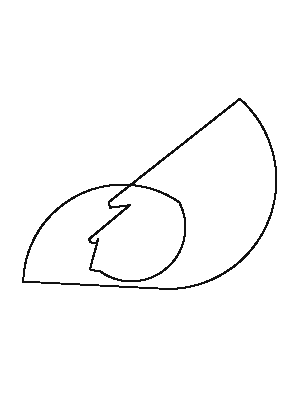}
\includegraphics[height=.24\textwidth]{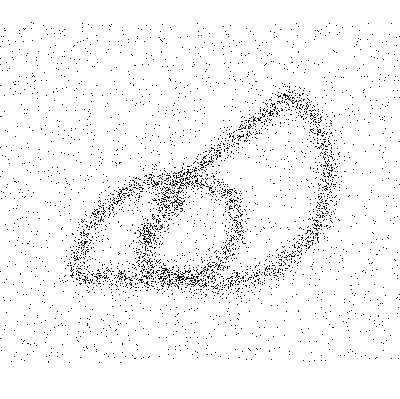}
\includegraphics[height=.24\textwidth]{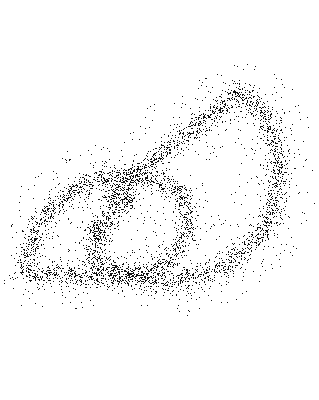}
\includegraphics[height=.24\textwidth]{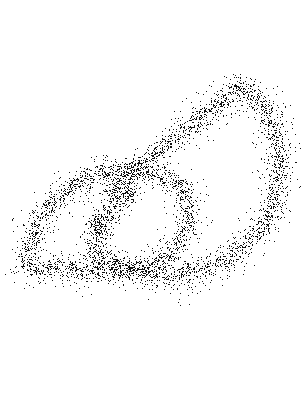}
\includegraphics[height=.24\textwidth]{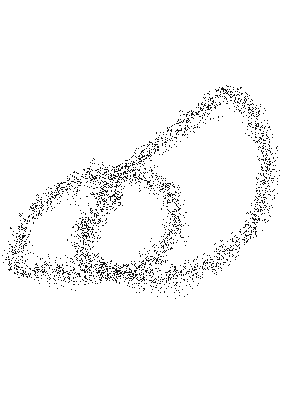}
\caption{Results of Algorithm {\sf ParfreeDeclutter} on two samples of one dimensional compact sets.
From left to right, the ground truth, the noisy input, two intermediate steps of the algorithm, and the final result.
}
\label{fig:1d}
\end{figure}

 Figure~\ref{fig:1dnotwork} gives an example of an adaptive sample where Algorithm {\sf ParfreeDeclutter} doesn't work.
 The ground truth is a heptagon with its vertices being connected to the center. 
Algorithm {\sf ParfreeDeclutter} doesn't work in this case because the sample is highly non-uniform and the ambient noise is very dense (63.33\% is ambient noise).
 The center part of the graph is significantly denser than other parts as the center vertex has a larger degree. 
 So the sparser parts (other edges of the star and heptagon) are regarded as noises by Algorithm {\sf ParfreeDeclutter} and thus removed.

 \begin{figure}[!ht]
 \centering
 \includegraphics[width=.24\textwidth]{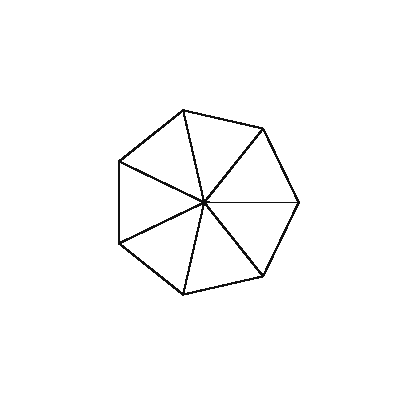}
 \includegraphics[width=.24\textwidth]{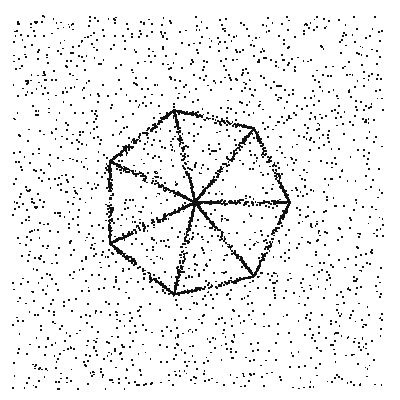}
 \includegraphics[width=.24\textwidth]{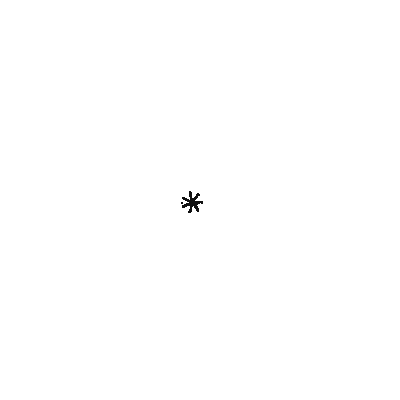}
 \includegraphics[width=.24\textwidth]{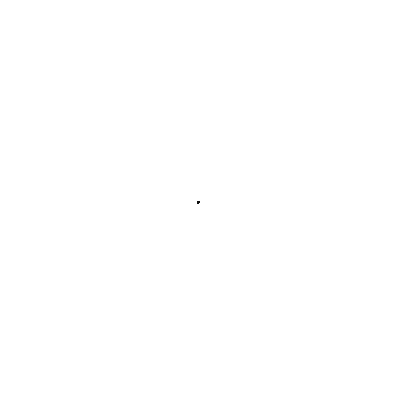}
 \caption{ A case where Algorithm {\sf ParfreeDeclutter} doesn't work.
 From left to right, the ground truth, the noisy input, an intermediate step of the algorithm and the final result.
 }
 \label{fig:1dnotwork}
 \end{figure}

Figure~\ref{fig:3dmanifold} (Figure \ref{fig:3dmanifoldone} in the main text) presents results obtained on an adaptive sample of a 2-manifold in 3D.
We consider again the so-called Botijo example with an adaptive sampling.
Contrary to the previous curves that were sampled uniformly, the density of this point set depends on the local feature size.
We also generate the noisy input the same way, adding a Gaussian noise at each point that has a standard deviation proportional to the local feature size.
Despite the absence of theoretical guarantees for the adaptive setting, Algorithm {\sf ParfreeDeclutter} removes the outliers while maintaining the points close to the ground truth.

\begin{figure}[!ht]
\centering
\includegraphics[height=.4\textwidth]{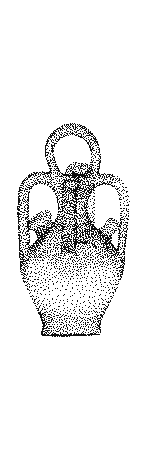}
\includegraphics[height=.4\textwidth]{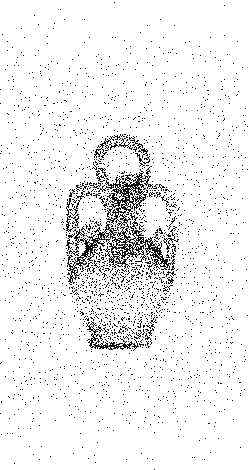}
\includegraphics[height=.4\textwidth]{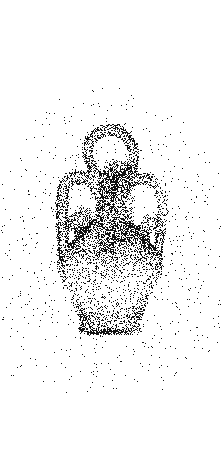}
\includegraphics[height=.4\textwidth]{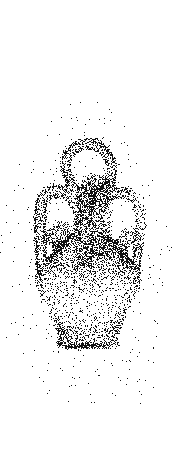}
\includegraphics[height=.4\textwidth]{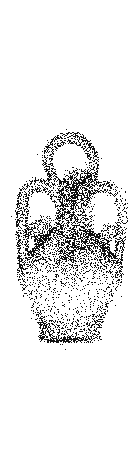}
\caption{Experiment on a two dimensional manifold.
From left to right, the ground truth, the noisy input, two intermediate steps of Algorithm {\sf ParfreeDeclutter} and the final result.
}
\label{fig:3dmanifold}
\end{figure}

\begin{table}[!ht]
\centering
\begin{tabular}{ c c c c c}
  \hline
  \hline
  Figure & Standard deviation of Gaussian & Size of ambient noise (percentage)  \\ \hline
  Figure~\ref{fig:1d} first row& 0.05 & 2000 (37.99\%)  \\
  Figure~\ref{fig:1d} second row& 0.05 & 2000 (45.43\%)  \\
  Figure~\ref{fig:3dmanifold} & 0.1 & 2000 (28.90\%)\\
  \hline
  \hline
\end{tabular}
\caption{Parameter of the noise model for Figure~\ref{fig:1d} and Figure~\ref{fig:3dmanifold}}
\label{tab:noise}
\end{table}

\begin{table}[!ht]
\centering
\begin{tabular}{ c c c c c c}
  \hline
  \hline
  Figure & Sample & Ground truth & Noise input & Intermediate steps & Final result  \\ \hline
  Figure~\ref{fig:1d} first row& uniform & 5264 & 7264 &6026 \hspace{.5cm} 5875& 5480 \\
  Figure~\ref{fig:1d} second row& uniform & 4402 & 6402 & 5197 \hspace{.5cm} 4992 & 4475 \\
  Figure~\ref{fig:3dmanifold}& adaptive & 6921 & 8921 & 7815 \hspace{.5cm} 7337 & 6983 \\
  \hline
  \hline
\end{tabular}
\caption{Cardinality of each dataset in Figure~\ref{fig:1d} and Figure~\ref{fig:3dmanifold}}
\label{tab:cardinality}
\end{table}

Finally, our last example is on a high dimensional data set.
We use subsets of the MINIST database.
This database contains handwritten digits.
We take all "1" digits (1000 images) and add some images from other digits to constitute the noise.
Every image is a $28\times28$ matrix and is considered as a point in dimension $784$.
We then use the $L_2$ metric between the images. 
Table~\ref{tab:highdimension} contains our experiment result.
Our algorithm partially removes the noisy points as well as a few good points.
If we add some random points in our space, we no longer encounter this problem,
which means if we add points with every pixel a random number, then we can remove all noises without removing any good points.

\begin{table}[!ht]
\centering
\begin{tabular}{ c c c c}
  \hline
  \hline
  Ground truth & Noise & Images removed after sampling & Digit 1 removed  \\ \hline
  1000 digit 1 & 200 digit 7 & 85 & 5 \\
  1000 digit 1 & 200 digit 8 & 94 & 5 \\
  1000 digit 1 & 200 digit 0-9 except 1 & 126 & 9 \\
  \hline
  \hline
\end{tabular}
\caption{Experiment on high-dimension datasets. The third and forth columns show number of corresponding images. }
\label{tab:highdimension}
\end{table}

\paragraph*{Denoising for data classification.} 
Finally, we apply our denoising algorithm as a pre-processing for high-dimensional data classification. Specifically, here we use MNIST data sets, which is a database of handwritten digits from '0' to '9'. 
Table \ref{tab:17cls-two} shows the experiment on digit 1 and digit 7. We take a random collection of 1352 images of digit '1' and 1279 images of digit '7' correctly labeled as training set, and take 10816 number of images   of digit 1 and digit 7  as testing set. Each of the image is $28 \times 28$ and thus can be viewed as a vector in $\mathbb{R}^{784}$. We use the $L_1$ metric to measure distance between such image-vectors. We use a linear SVM to classify the 10816 testing images. The classification error rate for the testing set is $0.6564\%$,  shown in the second row of Table \ref{tab:17cls-two}. 

Next, we artificially add two types of noises to input data: the {\it swapping-noise} and the {\it background-noise}. 

The swapping-noise means that we randomly mislabel some images of `1' as '7', and some images of `7' as '1'. As shown in the third row of Table \ref{tab:17cls-two}, the classification error increases to about $4.096\%$ after such mislabeling in the training set. 

Next, we apply our {\sf ParfreeDeclutter} algorithm to the training set (to the images with label '1' and those with label '7' separately) to first clean up the training set. As we can see in Row-6 of Table \ref{tab:17cls-two}, we removed most images with a mislabeled `1' which means the image is '7' but it is labeled as '1'. We then use the denoised dataset as the training set, and improved the classification error to $2.45\%$. 

While our denoising algorithm improved the classification accuracy, we note that it does not remove many mislabeled `7's from the set of images of digit `7'. The reason is that the images of `1' are significantly more clustered than those of digit `7'. Hence the set of images of `1's labelled as `7' themselves actually form a cluster; those points actually have even smaller k-distance than the images of `7' as shown in Figure \ref{fig:digits1and7}, and thus are considered to be signal by our denoising algorithm. 
\begin{figure}[H]
\centering
\includegraphics[height=0.3\textwidth]{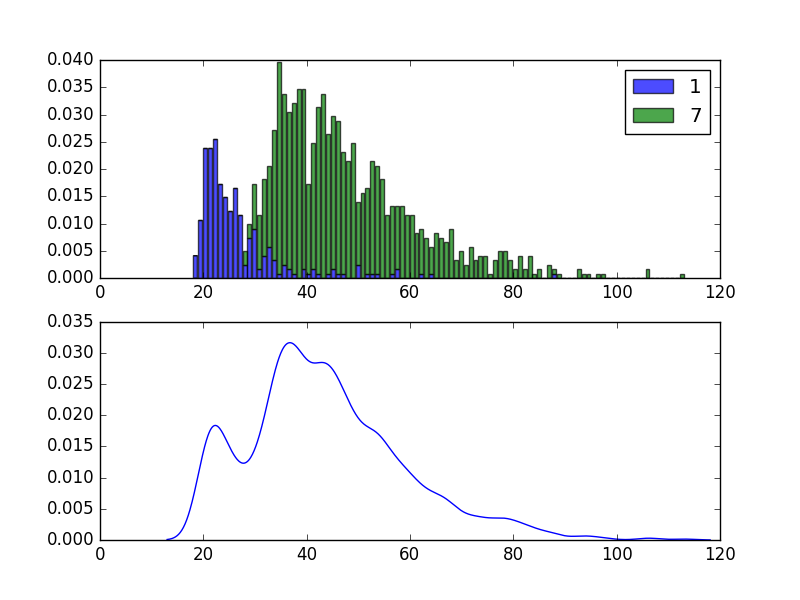}
\includegraphics[height=0.3\textwidth]{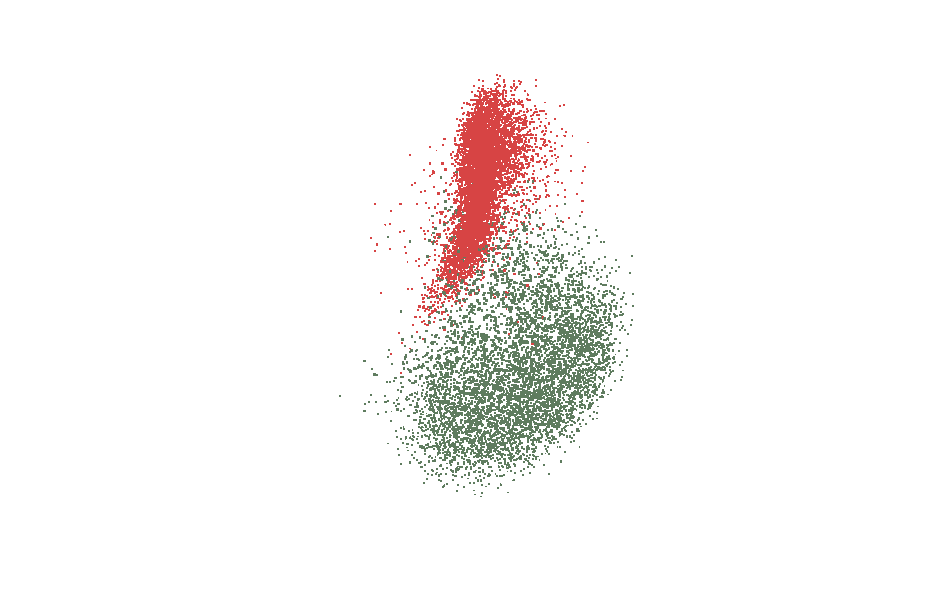}
\caption{Left: $k$-distance density distribution of digit 7s with digit 1s as the noise when $k=32$: Note, the points corresponding to images of digit `1' actually have a smaller $k$-distance than those of digit `7'. Right: Graph after being reduced to 3 dimension by PCA, the red one is for 1s, the green one is for 7s.}
\label{fig:digits1and7}
\end{figure}

The second type of noise is the {\it background noise}, where we replace the black backgrounds of a random subset of images in the training set (250 `1's and 250 `7's) with some other grey-scaled images. Under such noise, the classification error increases to $1.146\%$. Again, we perform our {\sf ParfreeDeclutter} algorithm to the training set separately, and use the denoised data sets as the new training set. The classification error is then improved to $0.7488\%$.

%
 \setcounter{magicrownumbers}{0} 
\begin{center}
\begin{table}[H]
\centering
\resizebox{\columnwidth}{!}{\hspace*{\fill}

    \begin{tabular}{|c |c | c | c | c | c | p{1.5cm} |}
    \hline
    \rownumber&\multicolumn{5}{c|}{}&Error(\%) \\ \hline
   \rownumber& Original & \multicolumn{2}{c|}{\# Digit 1\quad 1352} &   \multicolumn{2}{c|}{\# Digit 7\quad 1279}  &0.6564 \\ \hline 
 \multicolumn{7}{c}{}
   \\ \hline
      \rownumber&Swap. Noise& \multicolumn{2}{c|}{\# Mislabelled 1\quad 270} & \multicolumn{2}{c|}{\# Mislabelled 7\quad 266} &4.0957 \\ \hline

   \rownumber&&\multicolumn{2}{c|}{Digit 1}&\multicolumn{2}{c|}{Digit 7}& \\ \hline
    \rownumber&&\# Removed&\# True Noise&\# Removed&\# True Noise &\\ \hline
    \rownumber& L1 Denoising&314&264&17&1&2.4500\\  \hline

  \end{tabular}
  }
         \bigbreak
        \resizebox{\columnwidth}{!}{
        \begin{tabular}{|c| c | c | c | c | c | p{1.5cm} |}
        \hline
    \rownumber&Back. Noise& \multicolumn{2}{c|}{\# Noisy 1\quad 250} &\multicolumn{2}{c|}{\# Noisy 7\quad 250}  &1.1464 \\ \hline
    \rownumber& &\multicolumn{2}{c|}{Digit 1}&\multicolumn{2}{c|}{Digit 7}& \\ \hline
    \rownumber&&\# Removed&\# True Noise&\# Removed&\# True Noise &\\ \hline
     \rownumber&L1 Denoising&294&250&277&250&0.7488\\ \hline
    \end{tabular}
    }
    \caption{Results of denoising on digit 1 and digit 7 from the MNIST. }
    \label{tab:17cls-two}
  \end{table} 
\end{center}

%
%
%
%
%
%

\end{document}